\documentclass[journal, onecolumn, draftclsnofoot, 12pt]{IEEEtran}
\usepackage[utf8]{inputenc}

\usepackage[left=1in, right=1in, top=0.75in, bottom=1in]{geometry} 

\usepackage{amsmath,amssymb,amsfonts,amsthm}
\usepackage{algorithmic}
\usepackage{algorithm}
\usepackage{graphicx}
\usepackage{xcolor}
\usepackage{mathtools}
\usepackage{multirow}
\usepackage{caption}
\usepackage[font = footnotesize]{subfig}
\usepackage{tikzpagenodes}

\usepackage{wrapfig}

\definecolor{darkgreen}{rgb}{0.0, 0.75, 0.0}
\definecolor{orange}{rgb}{0.79, 0.38, 0.08}

\DeclareMathOperator*{\argmin}{argmin}

\usepackage[style=ieee, backend=biber]{biblatex}
\title{Knowledge-Driven Semantic Communication Enabled by the Geometry of Meaning}

\author{
    Dylan Wheeler, \IEEEmembership{Graduate Student Member, IEEE,} and Balasubramaniam Natarajan, \IEEEmembership{Senior Member, IEEE}
    \thanks{D. Wheeler and B. Natarajan are with the Mike Wiegers Department of Electrical and Computer Engineering at Kansas State University (email: dylan84@ksu.edu)}
}

\theoremstyle{definition}
\newtheorem{defn}{Definition}
\newtheorem{ex}{Example}

\newtheorem{lem}{Lemma}

\begin{document}

\maketitle



\begin{tikzpicture}[remember picture,overlay]
    \node[align=center, font=\scriptsize] at ([yshift=2em]current page text area.north) {
        \textcopyright 2023 IEEE.  Personal use of this material is permitted.  Permission from IEEE must be obtained for all other uses, in any current or future\\ media, including reprinting/republishing this material for advertising or promotional purposes, creating new collective works, for resale\\ or redistribution to servers or lists, or reuse of any copyrighted component of this work in other works.
    };
\end{tikzpicture}

\begin{abstract}
As our world grows increasingly connected and new technologies arise, global demands for data traffic continue to rise exponentially. Limited by the fundamental results of information theory, to meet these demands we are forced to either increase power or bandwidth usage. But what if there was a way to use these resources more efficiently? This question is the main driver behind the recent surge of interest in semantic communication, which seeks to leverage increased intelligence to move beyond the Shannon limit of technical communication. In this paper we expound a method of achieving semantic communication which utilizes the conceptual space model of knowledge representation. In contrast to other popular methods of semantic communication, our approach is intuitive, interpretable and efficient. We derive some preliminary results bounding the probability of semantic error under our framework, and show how our approach can serve as the underlying knowledge-driven foundation to higher-level intelligent systems. Taking inspiration from a metaverse application, we perform simulations to draw important insights about the proposed method and demonstrate how it can be used to achieve semantic communication with a 99.9\% reduction in rate as compared to a more traditional setup.
\end{abstract}

\begin{IEEEkeywords}
    semantic communication, conceptual spaces, cognitive communications, metaverse, 6G
\end{IEEEkeywords}

\section{Introduction}
\label{sec_intro}

In the modern age of connectivity, wireless networks are continuing to grow in both importance and ubiquity. As the number and size of these networks increases, so too does the overall demand for data traffic. A recent report from Ericsson \cite{ericsson_report} projects that global data traffic will rise exponentially to more than 400 EB/month by the end of 2028, from just over 100 EB/month in 2022. Current wireless communication paradigms are ill-suited to handle these growing demands due to the information-theoretic capacity of a wireless channel, also known as the \textit{Shannon rate} \cite{2006_coverThomas_infoTheory}. This fundamental limit provides two choices to increase the capacity of a wireless channel. One is to increase the power of the transmitted signal, which is not a sustainable option to meet exponentially rising demand. The other is to increase bandwidth, which has been the conventional approach with the push toward millimeter wave (mmWave) technology in fifth generation (5G) mobile networks \cite{3gpp_rel16} and recent research into terahertz (THz) communications \cite{2022_chaccour_THz}. However, operating at these extreme frequencies presents significant design challenges, such as extreme path loss \cite{2006_goldsmith_wirelessCommText}.

Due to these challenging circumstances, there has been a recent surge of interest in pushing beyond the Shannon rate. Consider the three levels of communication problem \cite{1949_shannonWeaver_mathTheoryOfComm}:
\begin{enumerate}
    \item[A.] How accurately can the symbols of communication be transmitted? (The technical problem.)
    \item[B.] How precisely do the transmitted symbols convey the desired meaning? (The semantic problem.)
    \item[C.] How effectively does the received meaning affect conduct in the desired way? (The effectiveness problem.)
\end{enumerate}
The Shannon rate provides a limit for the technical problem. On the other hand, the true objective of communication is rarely to simply transmit symbols without error. More often, one is trying to convey some meaning (level B) or affect some action (level C). Focusing on these goals, rather than solely on the technical goal, changes the fundamental meaning of channel capacity, because it changes what is considered to be an \textit{error}.

For example, consider the statement ``This communication is semantic" is transmitted, but the sentence ``This comnunication is semantic'' is received. A symbol error has occurred; however, the semantics of the statement will almost certainly be preserved, i.e., a semantic error does not occur. From this simple example, we arrive at the intuition that \textit{not all syntactic errors will induce a semantic error}, given sufficient background knowledge and reasoning capabilities. Thus, there is a potential to increase the effective capacity of wireless channels by operating at these higher levels of communication. This intuition, along with recent advances in artificial intelligence (AI), has led to the investigation of what is commonly referred to as \textit{semantic communication}.

Processing information at the semantic and effective levels comes with its own challenges, however. First and foremost, it forces us to reconsider fundamental notions, such as how we define and quantify information. In this article, we focus on the semantic level and draw upon insights from our previous work \cite{2023_wheeler_semComLetter} where we first introduced the use of Peter G\"ardenfors' theory of conceptual spaces \cite{2000_gardenfors_conceptSpace, 2014_gardenfors_semanticTheory} as a framework for modeling semantics. We introduce for the first time both theoretical and practical aspects of our approach, and demonstrate the potential of conceptual spaces as a knowledge representation technique which can serve as the foundation for reasoning-driven semantic communication systems.

\subsection{Related Work}
\label{subsec_rel_work}

The field of semantic communication has experienced a resurgence in the past few years. At the core of semantic communication is knowledge, and thus different approaches to semantic communication can naturally be partitioned by the mechanism used to model knowledge, or represent meaning. We briefly overview some of these approaches here, and more details can be found in some recent surveys of the field \cite{2023_wheeler_semComSurvey, 2022_iyer_semanticSurvey, 2021_lan_semanticSurvey}.

The first attempt to develop a theory of semantic communication was made by Carnap and Bar-Hillel \cite{1954_carnapBar-Hillel_theorySemInfo}. The key to their approach was to extend what we now know as information theory to include semantics by focusing on logical probabilities, rather than statistical probabilities. Some more recent works have attempted to build on this theory \cite{2011_bao_towardTheorySemComm, 2012_basu_qualityInfoSemanticRelationships}. This approach comes with known difficulties, such as the Bar-Hillel Carnap Paradox \cite{2004_floridi_theoryStrongSemInfo} and the challenge of building an expressive knowledge base using binary logic.

Others circumvent the challenge of representing meaning by focusing more on the effective level rather than the semantic level. More specifically, they consider the \textit{significance} of information \cite{2021_uysal_semCommNetworkSystems}, leading to \textit{goal-oriented} communication. ``Significance'' is with respect to some goal or task, and can be thought as ``provisioning of the right and significant piece of information to the right point of computation (or actuation) at the right point in time'' \cite{2021_uysal_semCommNetworkSystems}. A popular metric exemplifying this approach is the age of information (AoI) \cite{2021_uysal_aoiPractice}, which quantifies the ``freshness'' of data. Another is the so-called value of information \cite{2019_molin_valueOfInfo}.

One popular way of representing meaning is through the use of \textit{knowledge graphs}. This is the underlying modeling technique of the semantic web \cite{2009_hitzler_foundationsSemWeb}. In a knowledge graph, meaning is expressed through relations (edges) that connect entities (nodes) in the graph. The authors of \cite{2022_yang_semComEdgeIntel} propose an edge-sharing knowledge graph to enable semantic communication among intelligent devices in a network. In \cite{2023_kang_saliencyTaskOrientedSemCommImages}, the authors develop a communication system powered by a knowledge graph with subject-relation-object triples for transmission of images among unmanned aerial vehicles (UAVs).

The most popular approach to semantic communication has been through the use of machine learning (ML). ML-based approaches mostly handle semantics in an implicit manner, where meaning is learned by observation of massive data. Often, what these meanings are or what they represent is not apparent, due to the black-box nature of deep networks. Many of these approaches were inspired by the initial success of DeepSC \cite{2021_xie_deepSC} and its variants \cite{xie_deepsc_vqa, xie_lite_deepsc, xie_deepsc_speech}, which implement transformer-based architectures and achieve promising results \cite{2022_zhou_semCommAdaptiveTxfmr, 2022_zhou_adaptiveBitRateSemComm, 2023_wang_knowledgeEnhancedSemCom}. However, many works taking this approach (including DeepSC) train the system with the goal of exactly reproducing the initial data \cite{2022_wang_transformerEmpowered6gMimoSemCom, xie_lite_deepsc}, i.e., they are still operating at the technical level of communication.

The approaches mentioned above all have significant limitations when modeling semantic information. The symbolic approaches using logical probabilities and knowledge graphs face combinatorial challenges when used to model large knowledge bases. Significance-based approaches do not directly address the semantic level of communication, and ML-based approaches suffer from a lack of interpretability and explainability. Some approaches to semantic communication cut across these categories, such as the work in \cite{2023_gunduz_beyondBits}, which takes an information-theoretic view of task-oriented semantic communication, posing it as a rate distortion problem with context as side information. The authors then connect these theoretical perspectives with a ML-based approach, making the framework more practical.

Despite the many approaches discussed above, there is currently no consensus on how semantic information should be modeled and quantified. In our previous work \cite{2023_wheeler_semComLetter}, we introduced a novel approach based on the theory of conceptual spaces \cite{2000_gardenfors_conceptSpace, 2014_gardenfors_semanticTheory}. The initial treatment did not include any theoretical analysis and was limited to trivial examples. In this paper, we further develop the theoretical foundations of this approach and apply it to a practical metaverse-inspired problem to highlight its potential for future semantic communications.

\subsection{Contributions}
\label{subsec_contri}

In this work, we aim to formalize key theoretical aspects of a semantic communication framework utilizing conceptual spaces. In addition, we make a connection between our approach and a recently proposed framework for \textit{reasoning-driven} semantic communication \cite{2022_chaccour_LessDataMoreKnowledge}, and show how conceptual spaces can provide the knowledge-driven foundation for such communication. We provide thorough simulation results from a metaverse-inspired problem setup to confirm the theoretical developments and further demonstrate the potential of our approach for highly efficient communication. The main contributions of this work can be summarized as:
\begin{itemize}
    \item We uncover key theoretical aspects of the conceptual space-based framework for semantic communication. Specifically, we provide formal definitions of the various elements (domain, property, etc.) and examine the central concept of semantic distortion and semantic error under these definitions.
    \item We further analyze this notion of semantic distortion, and derive bounds on the probability of semantic error under our proposed framework. The derived bounds are simple to obtain under reasonable assumptions and are based on the level of semantic distortion introduced at the semantic encoder and the channel, and can thus be used to inform robust system design.
    \item We show how our approach can lay the knowledge-driven groundwork for a reasoning-driven semantic communication system, such as that proposed in \cite{2022_chaccour_LessDataMoreKnowledge}. Specifically, we demonstrate how the notion of a \textit{semantic language} \cite{2022_chaccour_LessDataMoreKnowledge} naturally arises from a conceptual space and how ideas like context can be easily incorporated into the conceptual space model. 
    \item We provide extensive simulation results of the proposed approach on a metaverse-inspired problem. We examine the problem of virtual reality exposure therapy (VRET) \cite{2022_bansal_healthcareMetaverse} and simulate a semantic communication system designed around this problem. 
    \begin{itemize}
        \item We simulate a system with a theoretical semantic encoder able to attain an arbitrary level of performance. We show that the end-end system performance is inevitably limited by the quality of the semantic encoder, and that the proposed system can be semantically robust under a fading channel model. Specifically, comparing performance under Rician fading to that of the AWGN channel, we find that in many cases an increase of only 2-3dB in signal-to-noise ratio (SNR) is needed under the fading channel model to achieve similar semantic performance to that of the AWGN model.
        \item A practical semantic communication system is simulated using a trained convolutional neural network (CNN) as the semantic encoder. We show that the proposed approach can reduce the communication rate by over 99.9\% as compared to a more traditional system, while achieving similar semantic performance at extreme SNR values. Additionally, the proposed technique exhibits significant performance gains in the mid-high SNR regions as compared to an alternative semantic communication technique, decreasing the probability of semantic error by around 0.25. Finally, we show that even greater gains can be achieved by fine-tuning the semantic encoder to the specific task at hand.
    \end{itemize}
\end{itemize}

The rest of this paper is organized as follows. In section \ref{sec_semComWithConSpaces}, we present our view on semantic communication. We then define the key elements of conceptual spaces, and introduce our model of a semantic communication system. We focus on the notion of semantic distortion in section \ref{sec_semDistAnalysis} and provide a formal definition based on our model. With this definition in hand, we define the notion of a semantic error, and derive two upper bounds on the probability of a semantic error. In section \ref{sec_knowDrivenCom}, we examine some extensions of our approach. First, we show how our approach can lay the knowledge-driven framework for higher-level reasoning-driven systems, and then we illustrate how the notion of context can be incorporated into the conceptual space model. We begin section \ref{sec_experiments} by first describing the methods used to test our proposed approach, and follow with our experimental results and discussion. Section \ref{sec_conclusion} concludes the paper.

\section{Foundations: Semantic Communication with Conceptual Spaces}
\label{sec_semComWithConSpaces}

In this section, we present our view on the meaning of ``semantic communication.'' We then provide the necessary background on the theory of conceptual spaces, formally defining core elements and laying the groundwork for later analysis. Finally, we provide our framework of semantic communication built on this groundwork.

\subsection{What is Semantic Communication?}
\label{subsec_whatIsSemCom}

\begin{wrapfigure}{l}{0.28\textwidth}
    \centering
    \includegraphics[width = 0.27\textwidth]{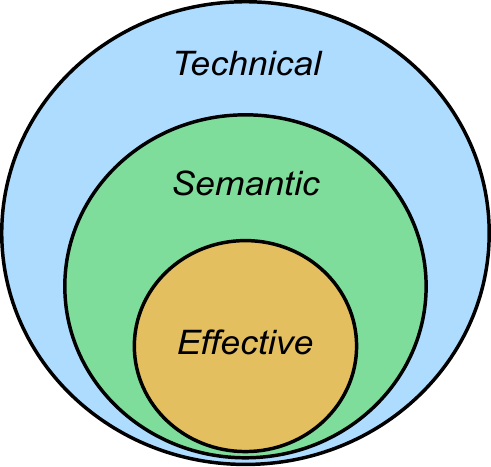}
    \caption{Hierarchy of communication goals}
    \label{fig_goalHierarchy}
\end{wrapfigure}

Lately, semantic communication has been a popular term used to describe a wide variety of approaches, and thus we will provide some clarification of what we mean when using this term here. First consider the hierarchy of communication goals shown in Figure \ref{fig_goalHierarchy}. Note that all communication must be technical; some kind of signal must be sent for communication to take place. A subset of these communication goals will also be semantic, where the sender wants to convey some meaning to the listener. An even smaller subset will also be effective, where the sender has a desired action they would like the listener to take. 

For example, suppose the sender transmits the text ``The oven is in the kitchen'' to the listener. If the goal is purely technical, then accurate transmission of each character is the sole objective. However, perhaps the sender is attempting to convey the layout of their home to the listener; then the meaning of the words is the focus, e.g., the listener should understand that the oven is in the kitchen and not the bedroom. But perhaps the sender wants the listener to turn the oven on; now the goal is effective as well.

When we refer to semantic communication, we mean communication problems where the goal falls into the semantic level of this hierarchy. Therefore, we do not mean using semantics to exactly reproduce the original data, as this falls into the technical level. Likewise, we do not mean goal-oriented or task-oriented communication, which sits at the effective level. In this work, we are concerned with problems where the primary goal is to convey some meaning. Thus, how meaning is defined and quantified is of prime importance.

\subsection{A Geometric View of Semantics}
\label{subsec_geometrySemantics}

\begin{wrapfigure}{r}{0.34\textwidth}
    \centering
    \includegraphics[width = 0.33\textwidth]{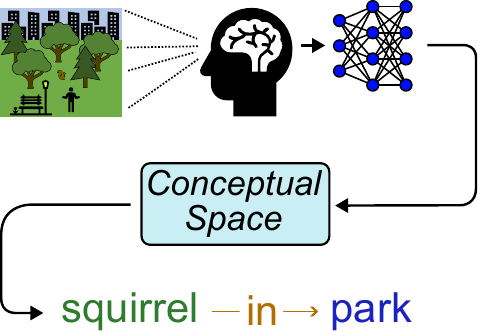}
    \caption{Visualization of the different levels of cognitive representation: associationism (top), conceptual (middle), and symbolic (bottom)}
    \label{fig_cogRepLevels}
\end{wrapfigure}

In \cite{2000_gardenfors_conceptSpace} Peter G\"ardenfors describes his theory of conceptual spaces, which is proposed as a cognitive model for how ideas and thoughts take shape in the human mind. There, G\"ardenfors defines three levels of cognitive representations. The \textit{associationism} level models cognitive representations as associations between different elements. He names \textit{connectionism} as a special case, where these associations are modeled by artificial neuron networks. In contrast, the \textit{symbolic} level uses discrete symbols as the primary cognitive representation and treats cognition as symbol manipulation. These two representations can be seen in the approaches to semantic communication described in subsection \ref{subsec_rel_work}. 

G\"ardenfors proposes a third level of representation, termed the \textit{conceptual} level, which acts as a bridge between these two. This kind of representation is ``based on geometrical structures, rather than symbols or connections among neurons'' \cite[p. 2]{2000_gardenfors_conceptSpace}. Rather than oppose each other, these different levels come together to form human thought. This idea is aligned with the recent interest in \textit{neurosymbolic AI} \cite{2023_garcez_neurosymbolicAI3rdWave}, and is shown in Figure \ref{fig_cogRepLevels}. The elements of the conceptual level of representation are originally described somewhat qualitatively \cite{2000_gardenfors_conceptSpace}, so we will now provide formal definitions for each of them.

The basic building block of a conceptual space is a \textit{quality dimension}, and is used to quantify some ``quality'' of an object or idea in some domain. We formalize this notion as a set.

\begin{defn}[Quality Dimension] \label{def_QualDim}
    A quality dimension is a set of scalar values quantifying some quality of an idea or object. We denote a quality dimension by $\mathcal{Q}$, where specific values along a quality dimension are denoted by $q$, such that $q \in \mathcal{Q}$.
\end{defn}

Quality dimensions are organized into \textit{domains}. Quality dimensions that share a domain are said to be \textit{integral}; G\"ardenfors describes integral quality dimensions as being those such that ``one cannot assign an object a value on one dimension without giving it a value on the other.'' \cite[p. 22]{2014_gardenfors_semanticTheory}. Thus, we can use quality dimensions to build a set representing the domain.

\begin{defn}[Domain] \label{def_Domain}
    A domain $\mathcal{D}$ is a set that is constructed from the Cartesian product of integral quality dimensions, i.e., $\mathcal{D} = \bigtimes_{k=1}^K Q_k$.  A point within a domain is specified by a column vector of quality values $\textbf{d} = \begin{pmatrix} q_1 & q_2 & \cdots & q_N \end{pmatrix}'$. 
\end{defn}

We use $'$ to denote the transpose operator. With a slight abuse of terminology, we refer to $K$ as the dimensionality of domain $\mathcal{D}$. Furthermore, we can think of ideas and objects as having \textit{properties}. These properties are with respect to a single domain.

\begin{defn}[Property] \label{def_Property}
    A property $\mathcal{P}$ of a single domain $\mathcal{D}$ is defined as a convex subset of that domain, i.e., $\mathcal{P} \subset \mathcal{D}$ where, for $\lambda \in [0,1]$ and any $\textbf{p}_1, \textbf{p}_2 \in \mathcal{P}$, $\lambda \textbf{p}_1 + (1-\lambda) \textbf{p}_2 \in \mathcal{P}$.
\end{defn}

The intuition underlying the convexity requirement is straightforward; if two objects possess a property, objects that are conceptually ``in between'' these two objects should also possess the same property. Allowing for multiple domains gives rise to a \textit{conceptual space}.

\begin{defn}[Conceptual Space] \label{def_ConSpace}
    A conceptual space $\mathcal{Z}$ is a set which is formed from the Cartesian product of $M \geq 1$ domains, i.e. $\mathcal{Z} = \bigtimes_{m=1}^M \mathcal{D}_m$. An idea or object is uniquely specified within the conceptual space by providing complete coordinates within every domain, i.e., $\textbf{z} = \begin{pmatrix} \textbf{d}_1' & \textbf{d}_2' & \cdots & \textbf{d}_M' \end{pmatrix}'$.
\end{defn}

Finally, we need to define a structure corresponding to the notion of a \textit{concept}. G\"ardenfors describes a concept as a collection of regions across domains throughout the conceptual space, as well as ``correlations between the regions from different domains'' \cite[p. 24]{2014_gardenfors_semanticTheory}. We will drop idea of correlation for now.

\begin{defn}[Concept] \label{def_Concept}
    A concept $\mathcal{C}$ is a region within the conceptual space spanning one or more domains. If contained within a single domain, we have $\mathcal{C} \subset \mathcal{D}$, and if the concept spans across $N$ domains, we have $\mathcal{C} \subset \bigtimes_{n=1}^N \mathcal{D}_n$. 
\end{defn}
A particular example of a concept is the Cartesian product of $N$ properties, $\mathcal{C} = \bigtimes_{n=1}^N \mathcal{P}_n$. We now present a simple example to clarify these definitions.

\begin{ex}[Colors and shapes] \label{ex_colorsShapes}
Consider the domain of colors. Studies indicate that humans perceive three dimensions regarding color: hue ($\mathcal{Q}_1$) , saturation ($\mathcal{Q}_2$), and brightness ($\mathcal{Q}_3$) \cite{2000_gardenfors_conceptSpace}. These are integral \textit{quality dimensions} for colors; a color cannot have a hue without also assigning it a saturation value. Thus, these dimensions form a \textit{domain}, which is shown in Figure \ref{fig_colorDomain}. Note that the hue dimension is circular; this will have certain implications when defining semantic distortion. Also, the dependence between the brightness and saturation dimensions that determines the spindle-like shape requires an additional condition beyond the simple Cartesian product, which would otherwise result in a cylinder. Such extra conditions can easily be incorporated into the prior definitions to extend the theory. 

Given the domain shown in Figure \ref{fig_colorDomain}, an example property might be \textit{light red}. This property $\mathcal{P}_1$ is the convex region in the top-half of the spindle with hue values in the range corresponding to red. Now consider the domain of \textit{regular polygons}, having one discrete dimension quantifying the number of sides of the polygon, i.e., $\mathcal{D} = \mathcal{Q}_4$, where $q_4 = 3$ corresponds to triangles, $q_4 = 4$ corresponds to squares, and so on. Furthermore, define the singleton set $\{4\}$ as the \textit{square} property $\mathcal{P}_2$. Then we can consider the Cartesian product of these two domains as the \textit{color-regular polygon conceptual space}, and the Cartesian product of $\mathcal{P}_1$ and $\mathcal{P}_2$ as the concept of a \textit{light red square}.
\end{ex}

\begin{wrapfigure}{l}{0.31\textwidth}
  \begin{center}
    \includegraphics[width=0.3\textwidth]{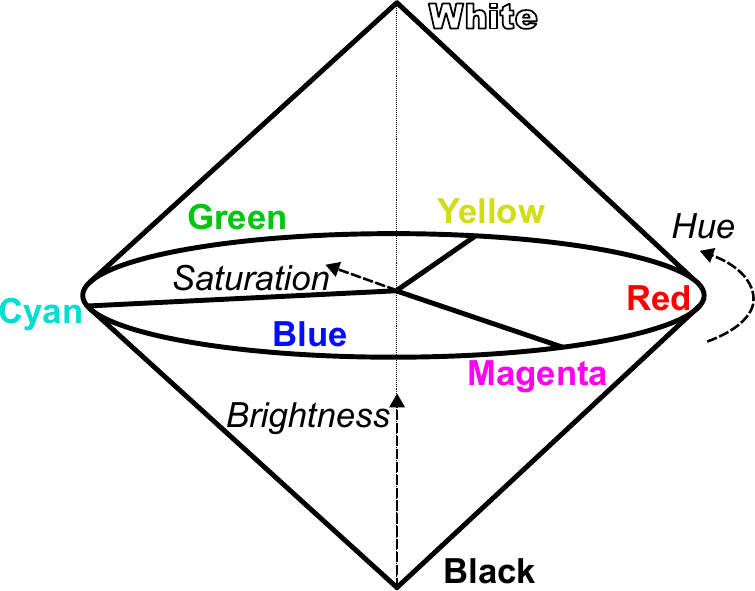}
  \end{center}
  \caption{\label{fig_colorDomain} Geometry of the domain from Example \ref{ex_colorsShapes}}
\end{wrapfigure}

Example \ref{ex_colorsShapes} shows how a simple conceptual space model can be designed to represent meaning. In the next subsection, we provide our general framework for semantic communication utilizing conceptual spaces for knowledge representation.

\subsection{Toward Semantic Communication}
\label{subsec_towardSemCom}

Figure \ref{fig_semComDiagram} provides a general block diagram of a framework for knowledge-driven semantic communication.
\begin{wrapfigure}{r}{0.39\textwidth}
  \begin{center}
    \includegraphics[width=0.382\textwidth]{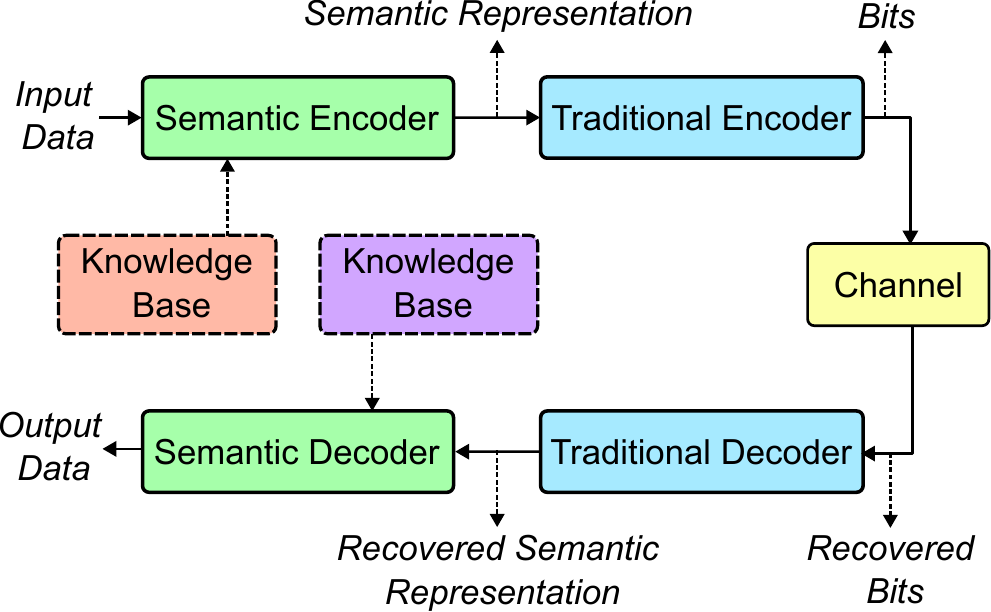}
  \end{center}
  \caption{\label{fig_semComDiagram} Block diagram of semantic communication}
\end{wrapfigure}
The proposed framework builds on existing communication paradigms with the addition of a semantic encoder and a semantic decoder. The primary task of the semantic encoder is to take some source data as an input and, according to its knowledge base, output a semantic representation that captures the meaning of the data. At the receiver, the semantic decoder is tasked with taking in a recovered (possibly corrupted) semantic representation and, according to its knowledge base, output some data.

Formally, the semantic encoder and decoder can be represented by functions. Let the input to the system be a vector $\textbf{x}$. Letting the semantic encoder with knowledge base $\mathcal{K}$ be denoted by the function $e_\mathcal{K}$, we have
\begin{equation} \label{eq_generalSemEnc}
    e_\mathcal{K}(\textbf{x}) = \textbf{z},
\end{equation}
where $\textbf{z}$ is the semantic representation for data $\textbf{x}$. Let the traditional encoder, channel, and traditional decoder be represented by functions $f$, $h$, and $g$, respectively. We can represent the traditional portion of the system as a composition of these functions
\begin{equation} \label{eq_tradFuncs}
    \hat{\textbf{z}} = g(h(f(\textbf{z}))).
\end{equation}
Finally, we represent the semantic decoder with knowledge base $\mathcal{L}$ as a function $d_\mathcal{L}$, and we have
\begin{equation} \label{eq_generalSemDec}
    d_\mathcal{L}(\hat{\textbf{z}}) = \textbf{u},
\end{equation}
where $\textbf{u}$ is the final output of the system, which can be some goal-driven transformation of $\hat{\textbf{z}}$.

We take the knowledge bases at the transmitter and the receiver to be conceptual space models. Moreover, for simplicity we assume that the transmitter and receiver share the same conceptual space $\mathcal{Z}$\footnote{Semantic communication with mismatched conceptual spaces at the transmitter and receiver is an interesting problem that is out of the scope of this paper, and will be considered in future work.}. Thus, we rewrite (\ref{eq_generalSemEnc}) and (\ref{eq_generalSemDec}) as
\begin{equation} \label{eq_semEnc}
    e_\mathcal{Z}(\textbf{x}) = \textbf{z},
\end{equation}
and
\begin{equation} \label{eq_semDec}
    d_\mathcal{Z}(\hat{\textbf{z}}) = \textbf{u}.
\end{equation}

The semantic decoder function (\ref{eq_semDec}) is a general function that represents a mapping to some goal-based output $\textbf{u}$. As our focus is on semantic communication, we consider a particular semantic decoder. Let the primary goal of communication be to successfully convey a \textit{concept} to the receiver, i.e., to convey meaning. Assuming that a finite number of concepts are defined over the conceptual space $\mathcal{Z}$, define a set $\mathcal{J}$ as the \textit{index} set of these concepts:
\begin{equation} \label{eq_indexSet}
    \mathcal{J} = \{1, 2, \ldots, J\}
\end{equation}
where $J$ is the total number of concepts defined over $\mathcal{Z}$, and $\jmath \in \mathcal{J}$ corresponds to concept $\mathcal{C}_\jmath$. Then we can consider the semantic decoder as a function $d_\mathcal{Z}: \mathcal{Z} \rightarrow \mathcal{J}$ defined as
\begin{equation} \label{eq_semDecIndex}
    d_\mathcal{Z}(\hat{\textbf{z}}) = \hat{\jmath}, \quad \hat{\jmath} \in \mathcal{J}.
\end{equation}
In the following section, we use these terms to define the notion of a semantic error.

\section{Semantic Distortion and Error Analysis}
\label{sec_semDistAnalysis}

With the conceptual space-based framework in place, we now turn to the problem of quantifying semantic distortion and defining semantic errors. In this section, we will show that using conceptual spaces to model knowledge leads to a notion of semantic distortion that is natural, intuitive, and interpretable. Furthermore, we use this notion to derive bounds of the probability of semantic error. We close this section by discussing some of the practical implications of this theory.

\subsection{Semantic Distortion in a Conceptual Space}
\label{subsec_semDistConSpace}

As we have defined our model of knowledge representation in terms of geometric structures, it is natural to think of semantic distortion as a \textit{distance measure}. Doing so gives the ability to compare the meanings of different objects and ideas. First, consider the distance between points within a single domain. For the $m$th domain $\mathcal{D}_m$, we define this distance as a function mapping two points in the domain to the non-negative real line, i.e. $\delta_m: \mathcal{D}_m \times \mathcal{D}_m \rightarrow \mathbb{R}_+$. While any general mapping can be considered, we are interested in functions that satisfy the three conditions of a metric, namely \textit{non-negativity}, \textit{symmetry}, and the \textit{triangle inequality}. Respectively,
\begin{align*}
    \delta_m(\textbf{d}_1, \textbf{d}_2) &\geq 0,\\
    \delta_m(\textbf{d}_1, \textbf{d}_2) &= \delta_m(\textbf{d}_2, \textbf{d}_1),\\
    \delta_m(\textbf{d}_1, \textbf{d}_3) &\leq \delta_m(\textbf{d}_1, \textbf{d}_2) + \delta_m(\textbf{d}_2, \textbf{d}_3).
\end{align*}

Given $\delta_m$ for each of the $M$ domains, we can now define semantic distortion.

\begin{defn}[Semantic Distortion]
    Semantic distortion is a function mapping two points in the conceptual space $\mathcal{Z}$ to the non-negative real line $\delta: \mathcal{Z} \times \mathcal{Z} \rightarrow \mathbb{R}_+$ given by
    \begin{equation} \label{eq_semanticDistortion}
        \delta(\textbf{z}_1, \textbf{z}_2) = \sum_{m=1}^M \delta_m(\textbf{z}_1, \textbf{z}_2)
    \end{equation}
    where $\delta_m(\textbf{z}_1, \textbf{z}_2) = \delta_m(\textbf{d}_1, \textbf{d}_2)$ is computed using the coordinates in $\textbf{z}_1, \textbf{z}_2$ corresponding to domain $\mathcal{D}_m$.
\end{defn}

\begin{ex}[Distance in the color domain]
    Recall the color domain of Example 1. To determine a distance function for this domain, we must take into account its specific geometry. If we let the hue dimension take values in the interval $[0,1]$, we cannot use absolute or squared difference to represent distance, since we need to capture the circularity. One way to represent distance on this dimension is with the expression
    \begin{equation} \label{eq_hueDist}
        \min(\vert h-\hat{h} \vert, \vert 1 - (h-\hat{h})\vert),
    \end{equation} \label{hueDistApprox}
    where $h$ represents the hue value. To avoid taking a minimum, (\ref{eq_hueDist}) can be approximated as
    \begin{equation} \label{eq_hueDistApprox}
        \gamma(h,\hat{h}) = -\frac{1}{\rho} \ln\left( \frac{1}{2}e^{-\rho \vert h - \hat{h}\vert} + \frac{1}{2} e^{-\rho(1-\vert h-\hat{h} \vert)} \right),
    \end{equation}
    with parameter $\rho > 0$. Since the brightness and saturation dimensions are linear, we can let them take values in the interval $[0,1]$ and use the squared difference to represent distance. Then we define distance within the domain as a mean squared error-like function given by
    \begin{equation}
        \delta_\text{color}(\textbf{d}, \hat{\textbf{d}}) = \frac{1}{3} \left( (s - \hat{s})^2 + (b - \hat{b})^2 + \gamma(h,\hat{h})^2 \right),
    \end{equation}
    where $s$ and $b$ are the saturation and brightness values, respectively, and $\textbf{d} = \begin{pmatrix} h & s & b \end{pmatrix}'$.
\end{ex}

Finally, we note that if each of the domain-specific distance functions satisfy the metric conditions, then the semantic distortion function given by (\ref{eq_semanticDistortion}) will also be a metric. Next, we'll use this observation to derive bounds on the probability of semantic error under some assumptions.

\subsection{Bounding the Probability of a Semantic Error}
\label{subsec_probBounds}

In traditional information theory, an error is characterized by the reception of an incorrect symbol. As discussed toward the end of section \ref{subsec_towardSemCom}, we are interested in communication when the primary goal is to convey a concept. Thus, using the notations of (\ref{eq_indexSet}) and (\ref{eq_semDecIndex}), we can define the notion of a \textit{semantic error}.
\begin{defn}[Semantic Error] \label{def_semError}
    A semantic error occurs when the concept decoded by the receiver does not match the concept that is to be conveyed, i.e.
    \begin{equation}
        d_\mathcal{Z}(\hat{\textbf{z}}) = \hat{\jmath} \neq \jmath^*,
    \end{equation}
    where $\jmath^*$ is the index of the true concept $\mathcal{C}_{\jmath^*}$ to be communicated.
\end{defn}

This definition of a semantic error is rather intuitive. If we take the meaning of an idea or object to be represented by its corresponding concept in a conceptual space, the \textit{successful conveyance of meaning} boils down to the successful communication of these concepts. This also matches the intuition that a syntactic error does not necessarily induce a semantic error. Perhaps the technical communication of $\textbf{z}$ contains errors, such that $\hat{\textbf{z}} \neq \textbf{z}$, i.e., the conceptual space coordinates obtained by the receiver have \textit{shifted} to a new location within the space. Since concepts are represented as regions, it is possible that $\hat{\textbf{z}}$ will be decoded as the correct concept.

To derive formal bounds on the probability of semantic error, we will make a few assumptions:
\begin{enumerate}
    \item[A1.] Concepts are formed from the Cartesian products of properties. Thus, concepts are \textit{convex} regions within the space, and each concept $\mathcal{C}_\jmath$ has a \textit{concept prototype} $\textbf{z}_\jmath$.
    \item[A2.] Semantic distortion $\delta: \mathcal{Z} \times \mathcal{Z} \rightarrow \mathbb{R}_+$ is a metric.
    \item[A3.] The semantic decoder is a minimum-distance decoder: $d_\mathcal{Z}(\hat{\textbf{z}}) = \argmin_{\jmath \in \mathcal{J}} \delta( \textbf{z}_\jmath, \hat{\textbf{z}}).$
\end{enumerate}

\begin{wrapfigure}{l}{0.33\textwidth}
  \begin{center}
    \includegraphics[width=0.32\textwidth]{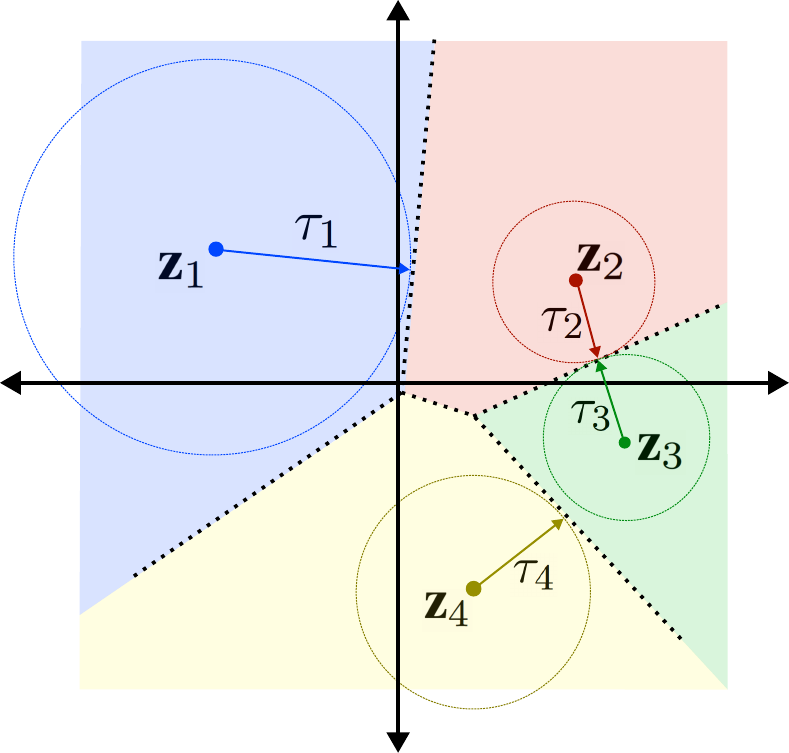}
  \end{center}
  \caption{\label{fig_voronoiTess} Example Voronoi tessellation with four concept prototypes on $\mathbb{R}^2$}
\end{wrapfigure}

Assumption A1 is intuitive, as a concept can often be expressed as a combination of properties, e.g., the concept of a ``red square'' is a combination of the properties ``red'' and ``square.'' Since properties are convex regions by definition, we can meaningfully compute a prototype point $\textbf{z}_\jmath$ as the centroid of these convex regions.  

Assumption A2 ensures that the chosen distortion function behaves like a true measure of distance. Given these first two assumptions, we can obtain a \textit{Voronoi tessellation} of the conceptual space \cite{2000_gardenfors_conceptSpace}. A Voronoi tessellation partitions a geometric space into convex regions, where the region of the tessellation corresponding to a prototype point $\textbf{z}_\jmath$ is defined as all points $\textbf{z}$ such that $\delta(\textbf{z}, \textbf{z}_\jmath) \leq \delta(\textbf{z}, \textbf{z}_j)$ for all $j \in \mathcal{J}$ where $j \neq \jmath$. Intuitively, it is the region of all points closer to the prototype of concept $\mathcal{C}_\jmath$ than any other prototype. An example is provided in Figure \ref{fig_voronoiTess}.

These first two assumptions lead naturally to a minimum-distance decoding scheme, expressed as the final assumption A3. Under this scheme, a concept will be decoded correctly if the received semantic representation $\hat{\textbf{z}}$ lies within the Voronoi region of the prototype point $\textbf{z}_{\jmath^*}$. Conversely, a semantic error will occur when $\hat{\textbf{z}}$ lies \textit{outside} this region. Let the Voronoi region corresponding to concept $\mathcal{C}_\jmath$ (equivalently, prototype $\textbf{z}_\jmath$) be denoted by $\mathcal{V}_\jmath$. Then we can express the probability of semantic error as in Definition \ref{def_semError} as
\begin{equation}
    P( \hat{\jmath} \neq \jmath^* ) = P( \hat{\textbf{z}} \notin \mathcal{V}_{\jmath^*} ).
\end{equation}

To obtain a bound on this probability, first consider the value $\tau_\jmath$ which we define as the solution to
\begin{align}
    \begin{split}
    \max_{\textbf{z} \in \mathcal{Z}} & \quad \delta( \textbf{z}, \textbf{z}_{\jmath} )\\
    \text{s.t.} & \quad \delta(\textbf{z},  \textbf{z}_\jmath) \leq \delta(\textbf{z}, \textbf{z}_j): \: j \in \mathcal{J}, \: j \neq \jmath.
    \end{split}
    \label{eq_tauOpt}
\end{align}
Intuitively, $\tau_\jmath$ is the radius of the largest sphere centered at $\textbf{z}_\jmath$ and inscribed in the Voronoi region $\mathcal{V}_\jmath$; this notion is illustrated in Figure \ref{fig_voronoiTess}. 

Next, consider the two potential sources of semantic distortion. Recall, the semantic encoder maps data $\textbf{x}$ to a semantic representation $\textbf{z}$. An imperfect encoder might introduce some level of semantic distortion $\delta(\textbf{z}_{\jmath^*}, \textbf{z})$. Moreover, syntactic errors may introduce some distortion $\delta(\textbf{z}, \hat{\textbf{z}})$. In general, the end-to-end distortion will be a function of these two values. Under assumption A2 that $\delta$ is a metric, the triangle inequality holds and we have
\begin{equation}
    \delta( \textbf{z}_{\jmath^*}, \hat{\textbf{z}}) \leq \delta(\textbf{z}_{\jmath^*}, \textbf{z}) + \delta(\textbf{z}, \hat{\textbf{z}})
\end{equation}

We can now derive an upper bound on the probability of semantic error.
\begin{lem} \label{lem_withPriors}
    Let $\alpha_\jmath$ denote the prior probability that concept $\mathcal{C}_\jmath$ is chosen as the true concept $\mathcal{C}_{\jmath^*}$. Then the probability of semantic error $\hat{\jmath} \neq \jmath^*$ has the upper bound
    \begin{equation} \label{eq_priorBound}
        P(\hat{\jmath} \neq \jmath^*) \leq \sum_{\jmath \in \mathcal{J}} \alpha_\jmath P\left(\delta(\textbf{z}_{\jmath}, \textbf{z}) + \delta(\textbf{z}, \hat{\textbf{z}}) > \tau_\jmath\right).
    \end{equation}
\end{lem}
\begin{proof}
    Recall that the Voronoi region corresponding to concept $\mathcal{C}_\jmath$ is denoted by $\mathcal{V}_\jmath$. Then
    \begin{align*}
        P(\text{sem. error}) 
        &= P( \hat{\jmath} \neq \jmath^* ) \\
        &= \sum_{\jmath \in \mathcal{J}} P( \hat{\jmath} \neq \jmath^* \mid \jmath^* = \jmath ) P( \jmath^* = \jmath)\\
        &= \sum_{\jmath \in \mathcal{J}} \alpha_\jmath P( \hat{\jmath} \notin \mathcal{V}_\jmath )\\
        &\leq \sum_{\jmath \in \mathcal{J}} \alpha_\jmath P( \delta(\hat{\textbf{z}},\textbf{z}_\jmath) > \tau_\jmath ) \\
        &\leq \sum_{\jmath \in \mathcal{J}} \alpha_\jmath P( \delta(\textbf{z}_\jmath, \textbf{z}) + \delta(\textbf{z},\hat{\textbf{z}}) > \tau_\jmath ) \qedhere
    \end{align*}
\end{proof}

Note that to compute (\ref{eq_priorBound}), the probability distribution over the set of concepts is required. If this distribution is unknown, with one additional assumption a looser bound can be derived.
\begin{lem} \label{lem_withoutPriors}
    Define $\tau = \min_{\jmath \in \mathcal{J}} \tau_\jmath$. If the distribution of the distortion introduced by the semantic encoder is independent of the concept $\mathcal{C}_\jmath$ being encoded, then the probability of semantic error has the upper bound
    \begin{equation} \label{eq_noPriorBound}
        P(\hat{\jmath} \neq \jmath^*) \leq P(\delta(\textbf{z}_{\jmath}, \textbf{z}) + \delta(\textbf{z}, \hat{\textbf{z}}) > \tau).
    \end{equation}
\end{lem}
\begin{proof} Beginning in the same manner as in the proof of Lemma \ref{lem_withPriors}, we arrive at the same result. Moreover, since $\tau \leq \tau_\jmath$ for all $\jmath \in \mathcal{J}$, we have
    \begin{align*}
        P(
        &\text{semantic error})\\ 
        &\leq \sum_{\jmath \in \mathcal{J}} P( \delta(\textbf{z}_\jmath, \textbf{z}) + \delta(\textbf{z},\hat{\textbf{z}}) > \tau_\jmath ) P(\jmath^* = \jmath)\\
        &\leq \sum_{\jmath \in \mathcal{J}} P( \delta(\textbf{z}_\jmath, \textbf{z}) + \delta(\textbf{z},\hat{\textbf{z}}) > \tau ) P(\jmath^* = \jmath)
    \end{align*}
    Looking at the first probability term inside the summation, we see that the only term dependent on $\jmath$ is $\delta(\textbf{z}_\jmath, \textbf{z})$. By assumption, this random distortion is independent of the concept $\mathcal{C}_\jmath$ being encoded, and thus the entire probability term does not depend on $\jmath$. Therefore, we can write
    \begin{align*}
        P(
        &\text{semantic error})\\ 
        &\leq  P( \delta(\textbf{z}_\jmath, \textbf{z}) + \delta(\textbf{z},\hat{\textbf{z}}) > \tau ) \sum_{\jmath \in \mathcal{J}} P(\jmath^* = \jmath)\\
        &=  P( \delta(\textbf{z}_\jmath, \textbf{z}) + \delta(\textbf{z},\hat{\textbf{z}}) > \tau ) \qedhere
    \end{align*}
\end{proof}

\subsection{Practical Implications}
\label{subsec_semDistImps}

The bounds derived in Lemmas \ref{lem_withPriors} and \ref{lem_withoutPriors} can guide the design of robust semantic communication systems. For example, in (\ref{eq_priorBound}) we see that the bound is determined by four values, namely, the prior distribution of concepts $\alpha_\jmath$, the quality of semantic encoding $\delta(\textbf{z}_{\jmath}, \textbf{z})$, the quality of syntactic communication $\delta(\textbf{z}, \hat{\textbf{z}})$, and structure of the conceptual space $\tau_\jmath$. Taking the $\alpha_\jmath$'s to be fixed, this leaves three variables that can be tuned to achieve the desired performance. First, the conceptual space should be designed such that common concepts (i.e., large $\alpha_\jmath$) have a large $\tau_j$, decreasing the probability of a misunderstanding when communicating that concept. After this, the traditional portion of the communication system can be determined, fixing the distribution of $\delta(\textbf{z}, \hat{\textbf{z}})$. Finally, the semantic encoder can be designed to meet performance requirements, for instance by training a deep neural network to meet the maximum distortion with high probability. Conversely, one can fix the semantic encoder and relax the requirements of the syntactic communication to achieve maximum technical efficiency while maintaining semantic performance. We end this section with an illustrative example of this process.

\begin{ex}[Semantic system design] \label{ex_systemDesign}
    Suppose we want to design a semantic communication system to communicate three concepts, denoted $\mathcal{C}_1$, $\mathcal{C}_2$, and $\mathcal{C}_3$. The probabilities of these concepts are known and fixed\footnote{Note that all values and distributions in this example are arbitrary, and are only meant to illustrate the process of semantic communication system design.} at $\alpha_1 = 0.5$, $\alpha_2 = 0.25$, and $\alpha_3 = 0.25$. After the conceptual space model is designed, it is found that $\tau_1 = 3$, $\tau_2 = 2$, and $\tau_3 = 1$. For simplicity, suppose that the components of semantic distortion introduced by both the encoder and the syntactic communication are independent of the concepts, and that they can be modelled as exponential random variables, i.e., $\delta_{enc} \sim \text{exp}(\lambda_1)$ and $\delta_{trad} \sim \text{exp}(\lambda_2)$. Parameters $\lambda_1$ and $\lambda_2$ reflect the capabilities of the system components, e.g., a greater $\lambda$ value indicates that large distortion values are less likely. After designing the semantic encoder, it is found that $\lambda_1  = 2$. We would like to design the syntactic portion of the system (channel coding, modulation, etc.) such that the probability that a semantic error occurs is less than 0.05. First, note that the sum of distortions follows a hypoexponential distribution. Letting $\delta = \delta_{enc} + \delta_{trad}$, the PDF of this random variable is given by the expression
    \begin{equation}
        p_{\delta}(x) = \frac{2\lambda_2}{2-\lambda_2}\left( e^{-\lambda_2 x}- e^{-2x} \right), \:\: x \geq 0.
    \end{equation}
    Integrating this PDF, we compute the probability of this sum exceeding some value $\tau_\jmath$ as
    \begin{equation}
        P(\delta > \tau_\jmath) = \int_{\tau_\jmath}^{\infty}p_\delta(x)dx = \frac{2e^{-\lambda_2\tau_\jmath} + \lambda_2 e^{-2\tau_\jmath}}{2 - \lambda_2}.
    \end{equation}
    Substituting this into the right hand side of (\ref{eq_priorBound}), setting equal to 0.05, and solving we find that $\lambda_2 \approx 1.5$. Designing the syntactic portion of the system to meet this specification will ensure that the desired performance is achieved without wasted resources. For $\lambda_2 < 1.5$, the distortion will be too great to guarantee the desired performance. However, if $\lambda_2 > 1.5$, resources will be wasted on unnecessary technical accuracy beyond what is required for the semantic goal.
\end{ex}

\section{Going Further: Knowledge-Driven Semantic Communication}
\label{sec_knowDrivenCom}

In \cite{2022_chaccour_LessDataMoreKnowledge}, the authors outline a vision and general framework for \textit{reasoning-driven} semantic communication. In this section, we briefly show how our framework based on conceptual spaces naturally lends itself as a knowledge-driven foundation to this kind of higher-level framework, where the notion of a ``semantic language'' as in \cite{2022_chaccour_LessDataMoreKnowledge} is elegantly realized by the conceptual space model. We then demonstrate how semantic distortion can be generalized to handle the notion of context and conclude with a discussion of semantic ambiguity, semantic redundancy and scalability as they relate to our framework.




\subsection{Building a Semantic Language}
\label{subsec_semLanguage}

First, consider the following definition for a \textit{semantic language}.
\begin{defn} (Semantic Language \cite[Defn. 3]{2022_chaccour_LessDataMoreKnowledge}) \label{def_semLanguage}
    A semantic language $\mathcal{L} = (X_{l_i}, Z_i)$, is a dictionary (from a data structure perspective) that maps the learnable data points $X_{l,i}$ to their corresponding semantic representation $Z_i$, based on the identified semantic content elements $Y_i$.
\end{defn}

We take the idea of a concept (as in Definition \ref{def_Concept}) to be synonymous with the idea of a ``semantic content element'' in Definition \ref{def_semLanguage}. Moreover, the geometric underpinning of the conceptual space model inherently provides the semantic representations of such content elements, which we can take to be the coordinates of the prototype point $\textbf{z}_\jmath$ for concept $\mathcal{C}_\jmath$. Then we can define a conceptual space-based semantic language.
\begin{defn} [Conceptual Space-Based Semantic Language] \label{def_ConSpaceSemLanguage}
    $\mathcal{L} = (\textbf{x}_{\jmath}, \textbf{z}_{\jmath})$ is a dictionary that maps the learnable data points $\textbf{x}_{\jmath}$ of concept $\mathcal{C}_\jmath$ to the corresponding concept prototype $\textbf{z}_{\jmath}$.
\end{defn} 

Therefore, our proposed model of knowledge representation is fully compatible with this higher-level framework for reasoning-driven communication. More importantly, our approach can provide a formal foundation for the qualitatively-defined components of this framework, such as semantic content elements and semantic representations.  In \cite{2022_chaccour_LessDataMoreKnowledge}, it is stated that a semantic language should exhibit three qualities. We close this subsection by highlighting how the semantic language of Definition \ref{def_ConSpaceSemLanguage} meets these three criteria:
\begin{itemize}
    \item \textit{Minimalism}: By representing concepts within the raw data as coordinates within a conceptual space, the amount of data transmitted to convey the semantics can be immensely reduced. For example, in our previous work we showed that effectively transmitting image semantics with conceptual spaces can reduce the required data rate by over 99\% \cite{2023_wheeler_semComLetter}.
    \item \textit{Generalizability}: As the conceptual space representation truly sits at the semantic level, domains can be reused for different tasks and settings. For example, the color domain can be used to represent concepts involving color in any given task.
    \item \textit{Efficiency}: The geometric representation of semantics provided by conceptual spaces allows for efficient computation of otherwise difficult operations. Take the example of semantic similarity, which has been a particularly challenging quantity to assess. One recent approach is to use a pre-trained DNN model to assess the similarity of sentences, e.g., using BERT \cite{2022_guo_signalShapingSemComm}. Even without training, a network of this size can still require billions of operations per computation. With conceptual spaces, semantic similarity is captured as a simple distance function, from which meaning can be easily compared.
\end{itemize}

\subsection{Communicating with Context}
\label{subsec_context}

When first describing conceptual spaces, G\"ardenfors points out the importance of context. He primarily describes two methods of expressing context in the model, which are referred to as \textit{salience} and \textit{sensitivity}. Salience refers to the fact that in a given situation, some domains of a conceptual space may be more important or prominent than others. To this point, G\"ardenfors states ``the relative weight of the domains depends on the \textit{context} in which the concept is used'' \cite[p. 103]{2000_gardenfors_conceptSpace}. On sensitivity, he states that ``subjects can be trained to become sensitized to certain \textit{areas} of a dimension so that the perceived length of the area is increased'' \cite[p. 104]{2000_gardenfors_conceptSpace}. 

These two aspects of context can be mathematically realized by altering the semantic distortion function in (\ref{eq_semanticDistortion}). First, note that these two notions will alter the semantic distortion function in different ways. Salience seems to work \textit{across} domains, while sensitivity alters the space \textit{within} a single domain. Starting with salience, as initially proposed by G\"ardenfors we can capture this idea with a set of \textit{weights} corresponding to the domains,
\begin{equation}
    \mathcal{W} = \{w_1, w_2, \ldots, w_M\}, \quad \sum_{m=1}^M w_m = 1,
\end{equation}
which alters the semantic distortion function in the follow way:
\begin{equation} \label{eq_semDistWithSalience}
    \delta(\textbf{z}_1, \textbf{z}_2 \mid \mathcal{W}) = \sum_{m=1}^M w_m \delta_m(\textbf{z}_1, \textbf{z}_2)
\end{equation}
Given the communication context, greater distortion in highly-weighted domains will have a greater effect of the overall distortion. Additionally, it is easy to show that if the semantic distortion $\delta(\cdot)$ is a metric function, the salience-weighted semantic distortion $\delta(\cdot \mid \mathcal{W})$ is also a metric, and the results in the previous section still apply.

Furthermore, we can think of sensitivity as a transformation of a domain, i.e., some dimensions may be ``streched'' and some may be ``shrunk'' to reflect the sensitivity of the agent in the given context. We define a set of transformations corresponding to each of the domains in the space,
\begin{equation}
    \mathcal{T} = \{T_1, T_2, \ldots, T_M\},
\end{equation}
such that
\begin{equation}
    \tilde{\mathcal{D}}_m = T_m(\mathcal{D}_m), \quad \tilde{\delta}_m = \delta_m(\tilde{\mathcal{D}}_m).
\end{equation}
Transformation $T_m$ alters the geometry of the $m$th domain to match the sensitivity of the agent in the context, and thus defines a new function $\tilde{\delta}_m$ which maps the new domain to the non-negative real line. Substituting these new distortion functions into (\ref{eq_semDistWithSalience}) we have
\begin{equation} \label{eq_semDistContext}
    \delta(\textbf{z}_1, \textbf{z}_2 \mid \mathcal{W}, \mathcal{T}) = \sum_{m=1}^M w_m \tilde{\delta}_m(\textbf{z}_1, \textbf{z}_2),
\end{equation}
which is the new context-dependent distortion function. By incorporating these different aspects, we can obtain a more realistic representation of meaning, and begin to capture some of the more nuanced aspects of meaning that can enable more intelligent semantic communication.

\subsection{Ambiguity, Redundancy, and Scalability} \label{subsec_ambigRedunScale}

Two important concepts in the context of semantic communication are semantic \textit{ambiguity} and \textit{redundancy}. Semantic ambiguity results from the fact that some data can be associated with multiple meanings. For example, the word ``tree'' in the context of a nature hike brings to mind an entirely different picture than that of a decision ``tree'' in a computer science context. Conversely, semantic redundancy refers to the scenario in which multiple unique data possess a similar meaning, e.g., the words ``begin'' and ``commence''. Semantic ambiguity can often be alleviated with the inclusion of contextual information. This is what is done in (\ref{eq_semDistContext}), where the contextual information is used to alter the way that meaning is represented within the conceptual spaces framework. Regarding the semantic redundancy, this issue is implicitly handled by the conceptual spaces framework in that data with similar meaning will be mapped to similar points within the semantic space. Thus, the elimination of semantic redundancy become of a problem of designing the semantic encoder to identify redundant syntactic data and map them accordingly to the semantic space. Our framework, therfore, provides natural ways of handling both semantic ambiguity and redundancy.

One of the primary challenges faced when implementing our proposed approach is the design of the conceptual space itself. This is a non-trivial task that generally requires deep knowledge of the semantics related to the task at hand. In our previous work \cite{2023_wheeler_semComLetter} and in the experiments in the following section, we hand-craft conceptual spaces to achieve semantic communication. This approach to conceptual space design lacks scalability, and thus more efficient techniques will be required to enable practical use of this framework. One possible path toward addressing this problem is to utilize existing techniques in the field of \textit{feature learning} or \textit{feature representation} \cite{2013_bengio_featureLearning} in machine learning. In this context, the dimensions of the conceptual space could be considered as the features to be learned, where the higher-level structures (domains, concepts, etc.) are then built on top of the learned space. We plan to study how feature learning can be leveraged to address this critical challenge in future work.

\section{Experimental Results}
\label{sec_experiments}

To demonstrate our proposed approach, we simulate a metaverse-inspired communication problem. Lately, the idea of the metaverse has received great attention due to its potential to enhance digital experiences. Metaverse applications aim to provide immersive online experiences by utilizing technologies such as virtual reality (VR) and augmented reality (AR) \cite{2021_ning_surveyMetaverse, 2022_sun_surveyMetaverse}. To achieve true immersion, these applications potentially require massive amounts of data transmission. These intensive communication requirements, coupled with the expected growth of the metaverse \cite{2023_kshetri_economicsIndustrialMetaverse}, make metaverse-inspired communication a prime candidate for the benefits of semantic communication.

\subsection{Problem Definition and System Description}
\label{subsec_probDef}

In our experiments, we look at the problem of virtual reality exposure therapy (VRET), which is just one of the many examples in which the metaverse can be utilized in the context of healthcare \cite{2022_bansal_healthcareMetaverse}. VRET can be used to treat phobias by exposing the patient to fear-inducing stimuli in the virtual environment while being guided by a therapist, to gradually reduce fear of the specific stimulus \cite{2022_albakri_vretReview}. One of the most common phobias is the fear of heights (acrophobia) \cite{2018_eaton_specificPhobias}. Thus, we focus on this scenario and aim to use the tools developed in the previous sections to design and analyze a semantic communication system for a VRET application within the metaverse. Specifically, we look at the end-end patient-to-therapist communication link.

\begin{wrapfigure}{l}{0.46\textwidth}
  \begin{center}
    \includegraphics[width=0.45\textwidth]{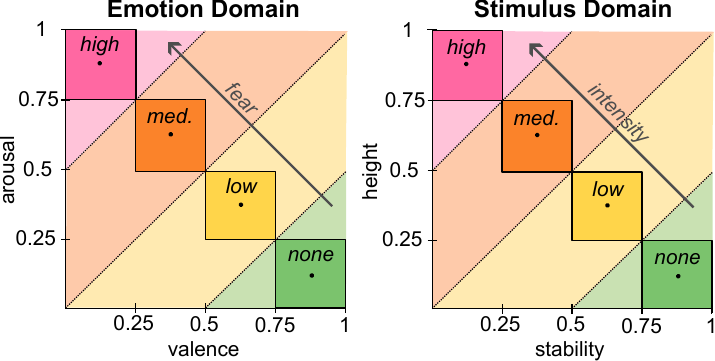}
  \end{center}
  \caption{\label{fig_vretCS} Conceptual space for using VRET to treat acrophobia. Properties, their corresponding prototype points and the Voronoi tesselation are also included.}
\end{wrapfigure}

The first step in our design is to build a conceptual space model to capture the relevant semantics. These relevant semantics are often tied to the overall goal of communication, and in the patient-to-therapist scenario, the goal is to provide a clear picture of the patients current emotional state to the therapist, as well as the intensity of the current stimulus in the VR environment. Thus, we define two domains, shown in Figure \ref{fig_vretCS}; namely, the emotion and stimulus domains.

We utilize valence and arousal dimensions \cite{2019_balan_fearLevelClassification} to construct the emotion domain. Valence refers to the degree to which a given emotion is positive or negative, while arousal refers to the intensity or level of activation of an emotion. Fear represents a low-valence and high-arousal emotion. We also adopt the different ranges used in \cite{2019_balan_fearLevelClassification} for fear-level classification to define properties within this domain, also shown in Figure \ref{fig_vretCS}. For the stimulus domain, with acrophobia in mind, we identify the dimensions of \textit{height} and \textit{stability} to characterize the \textit{intensity} of a stimulus. For example, a high-intensity stimulus would be climbing a tall ladder, while a low-intensity stimulus would be walking on a beach. For simplicity, properties are defined as degrees of intensity, analogous to the fear properties of the emotion domain.

As our conceptual space has two domains with linear dimensions, we can define the semantic distortion as the sum of the Euclidean distances in both domains, i.e., for $\textbf{z} = (v \:\:\: a \:\:\: h \:\:\: s)$,
\begin{equation} \label{eq_simsSemDist}
    \delta(\textbf{z}, \hat{\textbf{z}}) 
    = \left\Vert 
    \begin{pmatrix}
     v - \hat{v} \\ a - \hat{a}   
    \end{pmatrix}
    \right\Vert_2 
    + 
    \left\Vert 
    \begin{pmatrix}
     h - \hat{h} \\ s - \hat{s}   
    \end{pmatrix}
    \right\Vert_2
\end{equation}

Concepts are formed by taking the Cartesian product of convex regions within the two domains, and these concepts represent \textit{phobia levels}. For example, ``mild phobia'' is a concept obtained by combining ``mild fear'' with ``extreme intensity''. Similarly, a concept of ``extreme phobia'' can be realized through a combination of ``high fear'' and ``no intensity''. We define three concepts and their prototype points, which are given in Table \ref{tab_concepts}. We limit ourselves to concepts involving medium and high fear properties, due to these being by far the most common of our defined emotional properties present in the Aff-Wild2 dataset used for training \cite{2022_kollias_learningSyntheticData, 2019_kollias_deepAffectPrediction, 2019_balan_fearLevelClassification}.

\begin{wraptable}{r}{0.55\textwidth}
    \centering
    \small
    \begin{tabular}{|l|c|c|c|c|} 
        \hline
        \textit{Phobia Level} $\quad$ & \textbf{Valence} & \textbf{Arousal} & \textbf{Height} & \textbf{Stability} \\
        \hline
        \textbf{Mild} & 0.375 & 0.625 & 0.875 & 0.125 \\
        \hline
        \textbf{Moderate} & 0.250 & 0.750 & 0.500 & 0.500 \\
        \hline
        \textbf{Extreme} & 0.125 & 0.875 & 0.125 & 0.875 \\
        \hline
    \end{tabular}
    \caption{Concepts and Prototype Points $(\textbf{z}_{\jmath})$}
    \label{tab_concepts}
\end{wraptable}


As described in section \ref{sec_semComWithConSpaces}, our semantic communication goal is the accurate communication of these concepts. In our experiments, we simulate four systems to achieve this goal.

\subsubsection{Theoretical Semantic System}
\label{subsubsec_semSystemTheory}

First, we simulate semantic communication with a theoretical semantic encoder to examine the system performance with arbitrary encoder distortion. Here we model the dimensional error introduced by the semantic encoder as a multivariate normal (MVN) random vector, i.e.,
\begin{equation} \label{eq_semEncTheory}
    \textbf{z} = \textbf{z}_{\jmath^*} + \textbf{n},
\end{equation}
where $\textbf{n} \sim \mathcal{N}\left(\textbf{0}, \sigma_e^2 \textbf{I})\right)$ and $\textbf{I}$ denotes the identity matrix. Note that we can tune the parameter $\sigma_e$ to simulate varying levels of quality of the semantic encoder. The semantic representation $\textbf{z}$ is then quantized, modulated, and transmitted over the channel. At the receiver, demodulation and channel decoding is performed, and a minimum-distance semantic decoder is employed to make a decision on the transmitted concept.

\subsubsection{Practical Systems}
\label{subsubsec_practicalSystems}

We simulate three practical systems to demonstrate the performance of our proposed approach. One utilizes the conceptual space-based method of semantic communication described in previous sections, another uses the joint source channel coding (JSCC) technique for semantic communication, and the third is a traditional communication system that does not employ semantic communication. All three systems utilize a CNN architecture that is slightly modified from the one proposed in \cite{2019_handrich_simultaneousValenceArousal}, which is originally based on the YOLO architecture \cite{2016_redmon_YOLO}. The architecture of the CNN used for all three (termed BaseCNN) is shown in Figure \ref{fig_cnnArch}, and other details pertaining to the models are provided in Table \ref{tab_modelDetails}.

\begin{figure*}[t]
  \centering
  \includegraphics[width=0.95\textwidth]{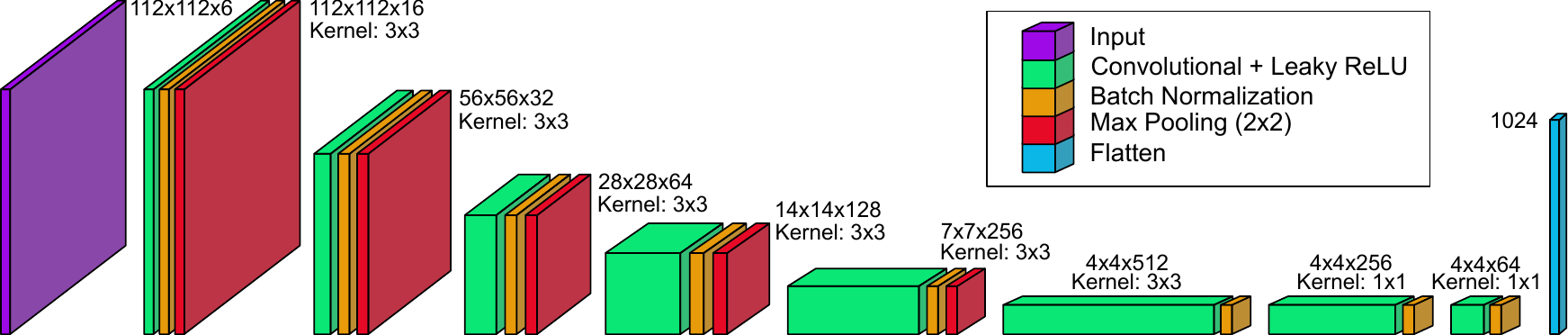}
  \caption{Visualization of the CNN (BaseCNN) architecture used in the three practical systems.}
  \label{fig_cnnArch}
\end{figure*}

\begin{table*}[t]
\centering
\renewcommand{\arraystretch}{0.6}
\begin{tabular}{|c|c|c|c|}
\cline{2-4}
\multicolumn{1}{c|}{} & \textbf{Conceptual Space-Based Semantic} & \textbf{End-End Semantic} & \textbf{Non-Semantic} \\
\hline
\multirow{6}{*}{\textit{Architecture}} & & BaseCNN & \\
 & & Dense(64), tanh activation & \\
 & BaseCNN & Binarization Layer & BaseCNN \\
 & Dense(4), sigmoid activation & 16-QAM Layer & Dense(3), softmax activation \\
 & & AWGN Layer & \\
 & & Dense (3), softmax activation & \\
\hline
\textit{Loss Function} & (\ref{eq_simsSemDist}) & \multicolumn{2}{c|}{Categorical Cross Entropy} \\
\hline
\textit{Optimizer} & \multicolumn{3}{c|}{Adam} \\
\hline
\textit{Learning Rate} & \multicolumn{3}{c|}{0.001} \\
\hline
\textit{Minibatch Size} & \multicolumn{3}{c|}{64} \\
\hline
\textit{Minibatches/Epoch} & \multicolumn{3}{c|}{500} \\
\hline
\textit{Total Epochs} & \multicolumn{3}{c|}{100} \\
\hline
\textit{Dropout Rate} & \multicolumn{3}{c|}{0.3} \\
\hline
\end{tabular}
\caption{Practical Systems: Models and Training Details}
\label{tab_modelDetails}
\end{table*}

To train the models, we utilized two datasets, each corresponding to one of the domains of the conceptual space. The first is the Aff-Wild2 dataset \cite{2022_kollias_learningSyntheticData, 2019_kollias_deepAffectPrediction}, which consists of 564 videos of around 2.8M frames with 554 subjects (326 of which are male and 228 female), all annotated with valence/arousal labels. For the stimulus domain, we created a dataset of 43 videos of around 500k frames of various stimuli corresponding to different levels of intensity and assigned each frame a height/stability label. Data samples were constructed by sampling a frame from each dataset, resizing each to a 112$\times$112$\times$3 RGB image, and combining the resulting images to form a 112$\times$112$\times$6 array which serves as the input to each CNN. A total of 450k of these combined images where used for the training dataset, while 50k were used for the validation dataset.

When training the conceptual-space based model, the label for a given input is the semantic representation \textbf{z}. For the other two models, the labels are one-hot encoded vectors specifying the present concept. For training the end-end semantic model, a uniformly random SNR ranging from -20dB to 20dB was selected for the AWGN layer for each data sample to generate the added noise, to avoid the need to retrain for each channel condition. For each model, dropout was applied during training after each convolutional block to mitigate overfitting. The models were created and trained using the keras API within the tensorflow Python package.

After the models were trained, the three systems were simulated as shown in the block diagrams in Figure \ref{fig_systemsDiagrams}. Specifically, the conceptual space-based system transmits semantic representations in a traditional manner and uses a minimum-distance decoder to obtain the concept. The end-end semantic system directly generates complex symbols, and uses a dense layer at the receiver to get the decision. The non-semantic system transmits both of the images over the channel and uses a classifier at the receiver to make a decision.

\begin{figure*}[t]
  \centering
  \includegraphics[width=0.88\textwidth]{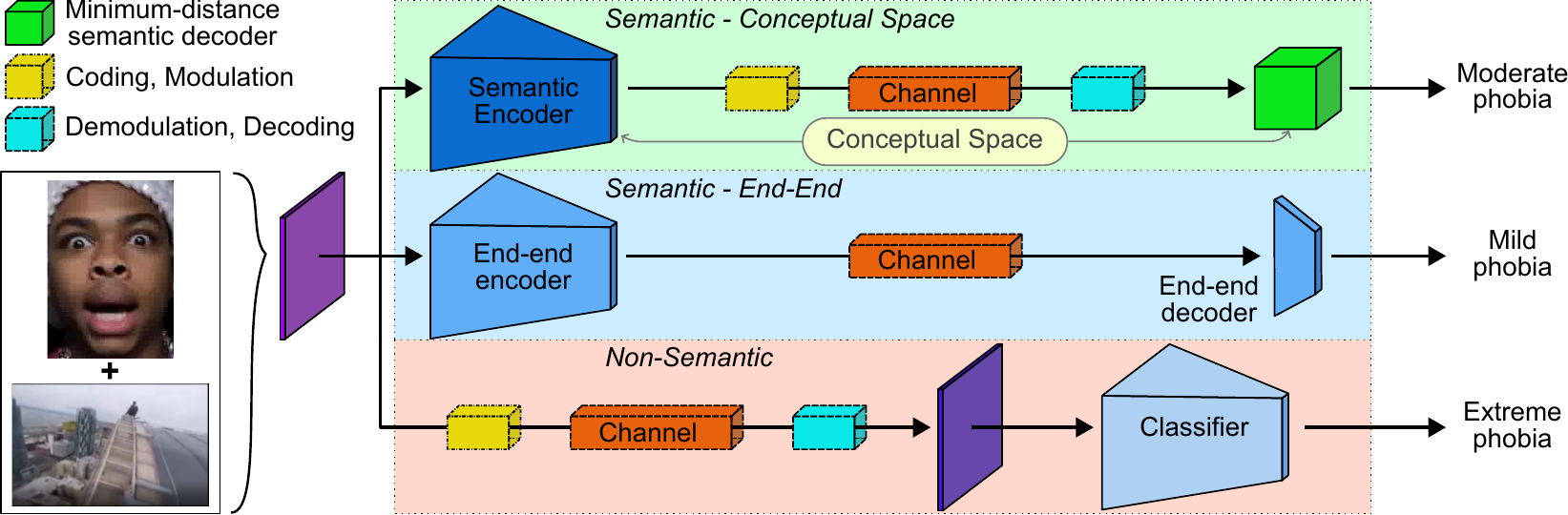}
  \caption{Diagrams of the three practical communication systems for VRET}
  \label{fig_systemsDiagrams}
\end{figure*}


\subsection{Traditional Communication Components}
\label{subsec_tradSystem}

For each of the systems, we simulate various aspects of the traditional portion of the system. Many future metaverse applications will likely utilize WiFi, and thus in our experiments we simulate the IEEE 802.11 standard \cite{2021_ieee802.11Standard}. We examine communication under BPSK, 16-QAM, and 256-QAM modulation. For the conceptual space-based semantic system, we employ rate 1/2 convolutional coding at the transmitter and Viterbi decoder at the receiver. Moreover, we perform experiments for both AWGN and Rician fading channel models. 


\begin{figure*}[t]
  \centering
  \includegraphics[width=1\textwidth]{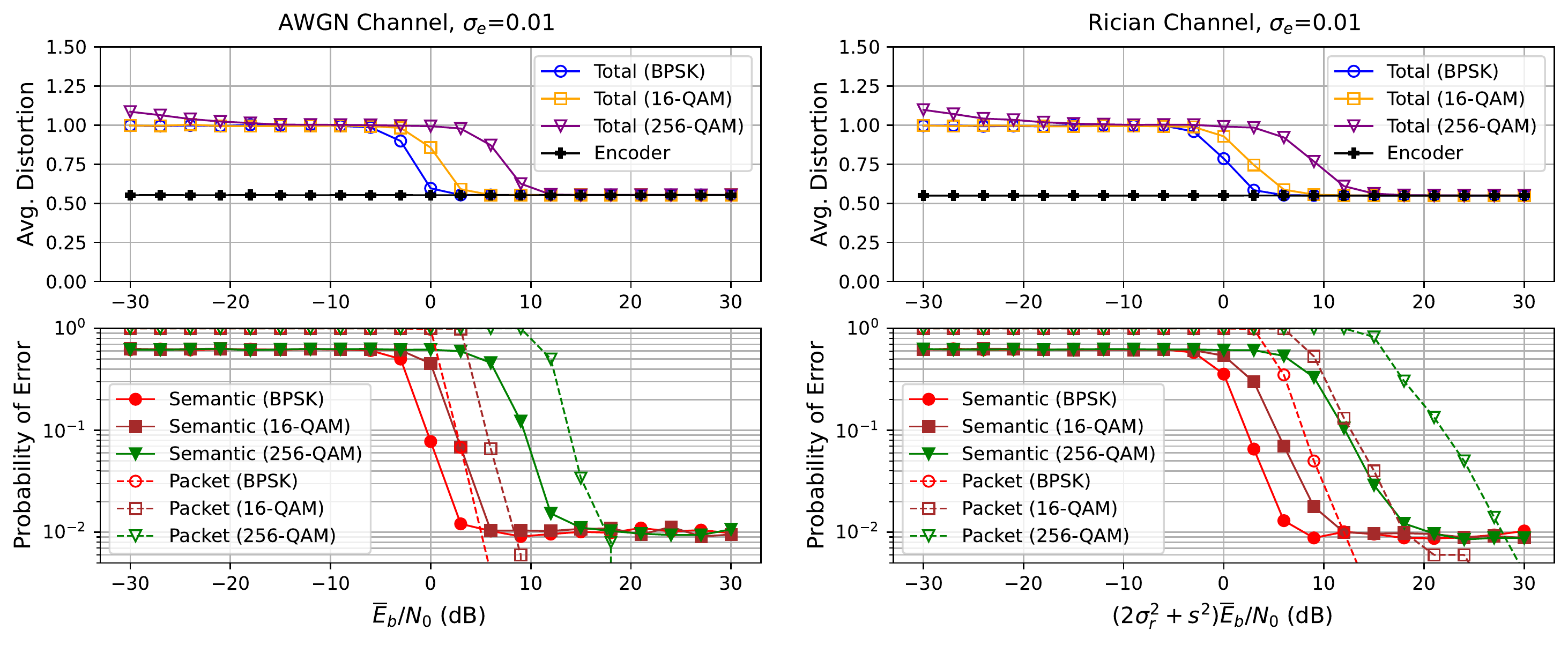}
  \caption{Plots of average semantic distortion (top) and probability of error (bottom) results for the semantic system with a theoretical semantic encoder under the AWGN channel (left) and Rician fading channel (right) plotted against average SNR}
  \label{fig_semTheoryResults}
\end{figure*}

\subsection{Results and Discussion}
\label{subsec_results}

Monte Carlo simulations are used to demonstrate the performance of the systems described above. Here we describe these experiments and the following results, and provide some discussion as to their implications.

\subsubsection{Theoretical Encoder}
\label{subsubsec_resultsTheoreticalEnc}

First, we perform experiments using the semantic system with a theoretical encoder, the results of which are shown in Figure \ref{fig_semTheoryResults}. The top-left plot shows the overall semantic distortion $\delta(\textbf{z}_{\jmath^*}, \hat{\textbf{z}})$ and the encoder distortion $\delta( \textbf{z}_{\jmath^*}, \textbf{z})$ as functions of the signal-to-noise ratio (SNR) for an encoder standard deviation $\sigma_e = 0.01$ under the AWGN channel. In each case, the overall semantic distortion becomes limited by the performance of the semantic encoder as the channel quality improves. Therefore, the semantic encoder is the key factor in determining the best-case performance of the overall semantic communication system. Improving the performance of the semantic encoder will lower the distortion floor of the overall system.

The bottom-left plot of Figure \ref{fig_semTheoryResults} confirms this; in these plots, the probability of semantic error $P(\hat{\jmath} \neq \jmath^*)$ is plotted against SNR, as well as the probability of packet error. In our simulations, a packet is composed of 160 bytes, which corresponds to 20 semantic representations $\textbf{z}$ each consisting of a vector of 4 8-bit floating point values, with rate 1/2 channel coding. We observe that the probability of semantic error reaches a hard limit at precisely the SNR where the overall distortion in the top-left plot reaches its floor. As such, this error floor can be reduced by improving the quality of the semantic encoder. This plot also illustrates the intuition that a syntactic error does not necessarily induce a semantic error. For example at $E_b/N_0 = 0$dB we see that the probability of packet error is nearly 1 for all modulation types, but for BPSK modulation the probability of semantic error is around $0.08$. Thus, we are able to achieve good semantic performance despite poor technical performance, illustrating the potential robustness of semantic communication systems to traditional bit errors.

The right-hand plots in Figure \ref{fig_semTheoryResults} provide similar results for the Rician channel model with fading parameter $K = 6$dB. We see that the semantic distortion curves are rather similar to those for the AWGN channel. Moreover, the probability of semantic error curves are also relatively similar to those for the AWGN channel, despite the greatly increased probability of packet error. These results indicate that the proposed approach can provide even greater benefits with respect to robustness in the presence of channel fading.

\begin{wrapfigure}{r}{0.49\textwidth}
  \begin{center}
    \includegraphics[width=0.48\textwidth]{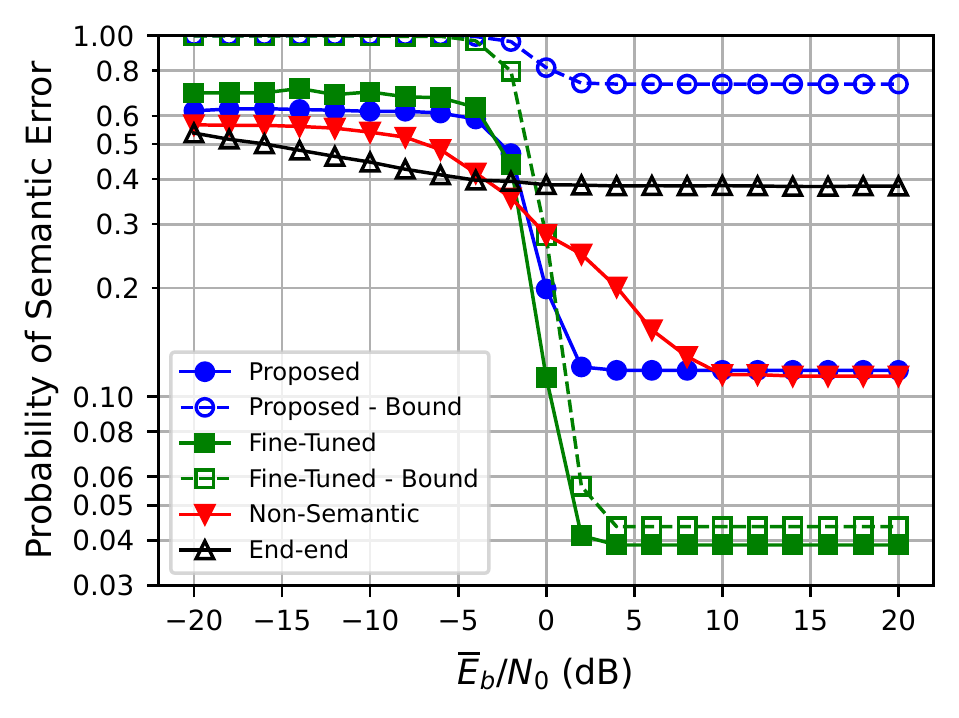}
  \end{center}
  \caption{\label{fig_semNonSemFineTunedPErrs} Probability of semantic error and the corresponding upper bound vs. average SNR for a general and fine-tuned semantic system, end-end semantic system, and non-semantic system, for 16-QAM modulation under the AWGN channel. The semantic systems achieve a rate reduction of over 99.9\% as compared to the non-semantic system.}
\end{wrapfigure}

\subsubsection{Practical Systems}
\label{subsubsec_resultsCNNSystems}

Next, we examine the performance of the semantic and non-semantic systems, and the results are shown in Figure \ref{fig_semNonSemFineTunedPErrs}. All of the results in this subsection were obtained for the AWGN channel.

The curves marked by circles show the probability of semantic error and the bound obtained from Lemma \ref{lem_withPriors}. We see similar behavior as the system with a theoretical semantic encoder; specifically, the performance improves with the channel conditions up to a certain point, at which it becomes limited by imperfections of the semantic encoder. We also observe that the upper bound provided by Lemma \ref{lem_withPriors} indeed holds as an upper bound for the semantic performance, though it does not appear to be a very tight bound. 

One reason for this loose bound lies in the fact that we are measuring semantic distortion with respect to concept prototypes. Therefore, even if the semantic encoder is perfect in mapping a given input to its true representation $\textbf{z}$, this point may not be close to any prototypes; in other words, some data are \textit{conceptually ambiguous} (which is different from \textit{semantic ambiguity} that was discussed in section \ref{subsec_ambigRedunScale}). Intuitively, we should observe improved performance and a tighter bound by eliminating this ambiguity. To test this hypothesis, we created a new dataset by randomly sampling points from the initial dataset that lie within the spheres centered at the prototype points with radius given by (\ref{eq_tauOpt}). This new dataset was then used to fine-tune the general CNN semantic encoder and to test the performance of the system with this fine-tuned encoder. The results are given in Figure \ref{fig_semNonSemFineTunedPErrs}, and are denoted by the curves marked with squares. These results confirm our intuition, as the best-case probability of error drops from around $0.11$ in the general case to around $0.04$. We also observe that the upper bounds on the semantic performance are indeed much tighter. 

The curve marked by filled triangles in Figure \ref{fig_semNonSemFineTunedPErrs} displays the results for the non-semantic system, while the unfilled triangles denote the performance of the end-end system.First, we observe that the end-end system performs the best out of the three for low SNR, which agrees with previous results for this kind of approach \cite{2021_xie_deepSC}. However, the end-end approach performs the worst at high SNR values. This is most likely due to the inclusion of the non-differentiable binarization, 16-QAM, and AWGN layers which introduce non-smooth behavior, impairing the gradient-based updates of the backpropagation algorithm during training. Next, we see that the non-semantic system performs similarly to the proposed system at extreme SNRs, with better performance at SNR values under 0dB and worse performance at SNR values greater than 0dB.

Finally, it is important to note the efficiency gains of the semantic systems to the non-semantic system. Each semantic representation consists of 4 8-bit values, coded at a rate of 1/2 for a total of 64 bits (the same number of bits is used in the end-end system). In the non-semantic system, each inference is carried out over 2 112$\times$112$\times$3, for a total of 75,264 pixels. Each pixel uses 8 bits, resulting in 602,112 bits transmitted per inference. Thus, the semantic systems reduce the overall rate by over 99.9\%, with the proposed system achieving the best performance. These results demonstrate that the proposed system can achieve reliable semantic communication with a massive reduction in rate.

\section{Conclusion}
\label{sec_conclusion}

In this paper, we have greatly expanded both theoretical and practical aspects of a semantic communication system with knowledge representation based on the theory of conceptual spaces. We have provided formal definitions of key aspects, defined the notions of semantic distortion and semantic error, and derived error bounds that follow from these definitions. We have shown how this theory can serve as the underlying foundation for a truly intelligent reasoning-driven system, and how important notions such as context can be easily incorporated into the framework. To illustrate some of the key benefits of our approach, we simulated a VRET application utilizing semantic communication, and demonstrated robust semantic performance with more than 99.9\% percent reduction in rate as compared to a more traditional system. The results reveal some important insights, such as the importance of the semantic encoder and the conceptual space design in the overall performance of the semantic communication system.

There are many interesting future directions that can be taken to extend this work. As was noted in section \ref{sec_semComWithConSpaces}, the study of semantic communication when the transmitter and receiver do not share a common conceptual space will be important to future development of these systems. Similarly, the process of autonomously learning the underlying conceptual space in an efficient and optimal way is another important area for future work. We plan to study this learning process, as well as how efficient and intelligent reasoning systems can be developed on top of a conceptual space-based knowledge base. Overall, we anticipate that the continued development and implementation of conceptual space-based semantic communication systems will unlock truly innovative and intelligent systems for the next generation of wireless communications.

\printbibliography

@techreport{ericsson_report,
    url = {https://www.ericsson.com/en/reports-and-papers/mobility-report/reports/november-2022},
    title = {{Ericsson} {Mobility} {Report}},
    institution = {Ericsson},
    month = {11},
    year = {2022}
}

@book{2006_coverThomas_infoTheory,
  at = {2008-03-31 06:17:47},
  author = {Cover, Thomas M. and Thomas, Joy A.},
  howpublished = {Hardcover},
  id = {1877660},
  isbn = {0471241954},
  month = {7},
  publisher = {Wiley-Interscience},
  title = {{Elements} of {Information} {Theory} {2nd} {Edition} ({Wiley} {Series} in {Telecommunications} and {Signal} {Processing})},
  year = 2006
}

@ARTICLE{2022_bansal_healthcareMetaverse,
  author={Bansal, Gaurang and Rajgopal, Karthik and Chamola, Vinay and Xiong, Zehui and Niyato, Dusit},
  journal={IEEE Access}, 
  title={{Healthcare} in {Metaverse}: {A} {Survey} on {Current} {Metaverse} {Applications} in {Healthcare}}, 
  year={2022},
  volume={10},
  number={},
  pages={119914-119946},
  doi={10.1109/ACCESS.2022.3219845}}

@article{3gpp_rel16,
 author = {3GPP},
 day = {24},
 institution = {{3rd Generation Partnership Project (3GPP)}},
 month = {06},
 note = {Version 16.0.0},
 number = {21.916},
 title = {{Release 16 Description; Summary of Rel-16 Work Items}},
 type = {Technical Report (TR)},
 url = {https://portal.3gpp.org/desktopmodules/Specifications/SpecificationDetails.aspx?specificationId=3493},
 year = {2021}
}

@ARTICLE{2022_chaccour_THz,
  author={Chaccour, Christina and Soorki, Mehdi Naderi and Saad, Walid and Bennis, Mehdi and Popovski, Petar and Debbah, Mérouane},
  journal={IEEE Communications Surveys \& Tutorials}, 
  title={{Seven Defining Features of Terahertz (THz) Wireless Systems: A Fellowship of Communication and Sensing}}, 
  year={2022},
  volume={24},
  number={2},
  pages={967-993},
  doi={10.1109/COMST.2022.3143454}}

@book{2006_goldsmith_wirelessCommText, place={Cambridge}, title={{Wireless Communications}}, DOI={10.1017/CBO9780511841224}, publisher={Cambridge University Press}, author={Goldsmith, Andrea}, year={2005}}

@book{1949_shannonWeaver_mathTheoryOfComm,
  address = {Urbana, IL},
  author = {Shannon, Claude E. and Weaver, Warren},
  description = {1998 Reprint of the First Paperback Edition 1963},
  isbn = {978-0-252-72548-7},
  publisher = {University of Illinois Press},
  title = {{The Mathematical Theory of Communication}},
  username = {flint63},
  year = 1949
}

@ARTICLE{2023_wheeler_semComLetter,
  author={Wheeler, Dylan and Tripp, Erin E. and Natarajan, Balasubramaniam},
  journal={IEEE Communications Letters}, 
  title={{Semantic Communication With Conceptual Spaces}}, 
  year={2023},
  volume={27},
  number={2},
  pages={532-535},
  doi={10.1109/LCOMM.2022.3230246}}

@book{2000_gardenfors_conceptSpace,
    author = {Peter G\"ardenfors},
    year    = {2000},
    publisher = {Massachusetts Institute of Technology},
    title = {{Conceptual Spaces: The Geometry of Thought}}
}

@book{2014_gardenfors_semanticTheory,
    author = {Peter G\"ardenfors},
    year    = {2014},
    publisher = {Massachusetts Institute of Technology},
    title = {{The Geometry of Meaning: Semantics Based on Conceptual Spaces}}
}

@ARTICLE{2023_wheeler_semComSurvey,
  author={Wheeler, Dylan and Natarajan, Balasubramaniam},
  journal={IEEE Access}, 
  title={{Engineering Semantic Communication: A Survey}}, 
  year={2023},
  volume={11},
  number={},
  pages={13965-13995},
  doi={10.1109/ACCESS.2023.3243065}}

@article{2022_iyer_semanticSurvey,
	doi = {10.1007/s11277-022-10111-7},
	year = 2022,
	month = {11},
	publisher = {Springer Science and Business Media {LLC}},
	volume = {129},
	number = {1},
	pages = {569--611},
	author = {Sridhar Iyer and Rajashri Khanai and Dattaprasad Torse and Rahul Jashvantbhai Pandya and Khaled M. Rabie and Krishna Pai and Wali Ullah Khan and Zubair Fadlullah},
	title = {{A Survey on Semantic Communications for Intelligent Wireless Networks}},
	journal = {Wireless Personal Communications}
}

@ARTICLE{2021_lan_semanticSurvey,
  author={Lan, Qiao and Wen, Dingzhu and Zhang, Zezhong and Zeng, Qunsong and Chen, Xu and Popovski, Petar and Huang, Kaibin},
  journal={Journal of Communications and Information Networks}, 
  title={{What is Semantic Communication? A View on Conveying Meaning in the Era of Machine Intelligence}}, 
  year={2021},
  volume={6},
  number={4},
  pages={336-371},
  doi={10.23919/JCIN.2021.9663101}}

@article{1954_carnapBar-Hillel_theorySemInfo,
	number = {3},
	author = {Rudolf Carnap and Yehoshua Bar{-}Hillel},
	year = {1954},
	volume = {19},
	pages = {230--232},
	publisher = {Association for Symbolic Logic},
	journal = {Journal of Symbolic Logic},
	title = {{An Outline of a Theory of Semantic Information}},
	doi = {10.2307/2268645}
}

@INPROCEEDINGS{2011_bao_towardTheorySemComm,  author={Bao, Jie and Basu, Prithwish and Dean, Mike and Partridge, Craig and Swami, Ananthram and Leland, Will and Hendler, James A.},  booktitle={2011 IEEE Network Science Workshop},   title={{Towards a Theory of Semantic Communication}},   year={2011},  volume={},  number={},  pages={110-117},  doi={10.1109/NSW.2011.6004632}}

@INPROCEEDINGS{2012_basu_qualityInfoSemanticRelationships,  author={Basu, Prithwish and Bao, Jie and Dean, Mike and Hendler, James},  booktitle={2012 IEEE International Conference on Pervasive Computing and Communications Workshops},   title={{Preserving Quality of Information by Using Semantic Relationships}},   year={2012},  volume={},  number={},  pages={58-63},  doi={10.1109/PerComW.2012.6197583}}

@article{2004_floridi_theoryStrongSemInfo,
author = {Floridi, Luciano},
year = {2004},
month = {05},
pages = {},
title = {{Outline of a Theory of Strongly Semantic Information}},
volume = {14},
isbn = {9780199232383},
journal = {Minds and Machines},
doi = {10.1023/B:MIND.0000021684.50925.c9}
}

@article{2021_uysal_semCommNetworkSystems,
      title={{Semantic Communications in Networked Systems}}, 
      author={Elif Uysal and Onur Kaya and Anthony Ephremides and James Gross and Marian Codreanu and Petar Popovski and Mohamad Assaad and Gianluigi Liva and Andrea Munari and Touraj Soleymani and Beatriz Soret and Karl Henrik Johansson},
      year={2021},
      eprint={2103.05391},
      archivePrefix={arXiv},
      primaryClass={eess.SP}
}

@article{2021_uysal_aoiPractice,
  doi = {10.48550/ARXIV.2106.02491},
  url = {https://arxiv.org/abs/2106.02491},
  author = {Uysal, Elif and Kaya, Onur and Baghaee, Sajjad and Beytur, Hasan Burhan},
  title = {{Age of Information in Practice}},
  publisher = {arXiv},
  year = {2021},
  copyright = {arXiv.org perpetual, non-exclusive license}
}

@article{2019_molin_valueOfInfo,
title = {{Scheduling Networked State Estimators Based on Value of Information}},
journal = {Automatica},
volume = {110},
pages = {108578},
year = {2019},
issn = {0005-1098},
doi = {https://doi.org/10.1016/j.automatica.2019.108578},
author = {Adam Molin and Hasan Esen and Karl Henrik Johansson}
}

@book{2009_hitzler_foundationsSemWeb,
  title={{Foundations of Semantic Web Technologies}},
  author={Hitzler, P. and Krotzsch, M. and Rudolph, S.},
  isbn={9781420090512},
  series={Chapman \& Hall/CRC Textbooks in Computing},
  year={2009},
  publisher={CRC Press}
}

@ARTICLE{2022_yang_semComEdgeIntel,
  author={Yang, Wanting and Liew, Zi Qin and Lim, Wei Yang Bryan and Xiong, Zehui and Niyato, Dusit and Chi, Xuefen and Cao, Xianbin and Letaief, Khaled B.},
  journal={IEEE Wireless Communications}, 
  title={{Semantic Communication Meets Edge Intelligence}}, 
  year={2022},
  volume={29},
  number={5},
  pages={28-35},
  doi={10.1109/MWC.004.2200050}}

@ARTICLE{2023_kang_saliencyTaskOrientedSemCommImages,
  author={Kang, Jiawen and Du, Hongyang and Li, Zonghang and Xiong, Zehui and Ma, Shiyao and Niyato, Dusit and Li, Yuan},
  journal={IEEE Journal on Selected Areas in Communications}, 
  title={{Personalized Saliency in Task-Oriented Semantic Communications: Image Transmission and Performance Analysis}}, 
  year={2023},
  volume={41},
  number={1},
  pages={186-201},
  doi={10.1109/JSAC.2022.3221990}}

@ARTICLE{2021_xie_deepSC,  author={Xie, Huiqiang and Qin, Zhijin and Li, Geoffrey Ye and Juang, Biing-Hwang},  journal={IEEE Transactions on Signal Processing},   title={{Deep Learning Enabled Semantic Communication Systems}},   year={2021},  volume={69},  number={},  pages={2663-2675},  doi={10.1109/TSP.2021.3071210}}

@INPROCEEDINGS{xie_deepsc_speech,  author={Weng, Zhenzi and Qin, Zhijin and Li, Geoffrey Ye},  booktitle={ICC 2021 - IEEE International Conference on Communications},   title={{Semantic Communications for Speech Signals}},   year={2021},  volume={},  number={},  pages={1-6},  doi={10.1109/ICC42927.2021.9500590}}

@ARTICLE{xie_deepsc_vqa,  author={Xie, Huiqiang and Qin, Zhijin and Li, Geoffrey Ye},  journal={IEEE Wireless Communications Letters},   title={{Task-Oriented Multi-User Semantic Communications for VQA}},   year={2022},  volume={11},  number={3},  pages={553-557},  doi={10.1109/LWC.2021.3136045}}

@ARTICLE{xie_lite_deepsc,  author={Xie, Huiqiang and Qin, Zhijin},  journal={IEEE Journal on Selected Areas in Communications},   title={{A Lite Distributed Semantic Communication System for Internet of Things}},   year={2021},  volume={39},  number={1},  pages={142-153},  doi={10.1109/JSAC.2020.3036968}}

@article{2022_zhou_semCommAdaptiveTxfmr,
	title = {Semantic {Communication} {With} {Adaptive} {Universal} {Transformer}},
	volume = {11},
	issn = {2162-2345},
	doi = {10.1109/LWC.2021.3132067},
	journal = {IEEE Wireless Communications Letters},
	author = {Zhou, Qingyang and Li, Rongpeng and Zhao, Zhifeng and Peng, Chenghui and Zhang, Honggang},
	month = {3},
	year = {2022},
	note = {Conference Name: IEEE Wireless Communications Letters},
	pages = {453--457},
}

@article{2022_zhou_adaptiveBitRateSemComm,
	title = {Adaptive {Bit} {Rate} {Control} in {Semantic} {Communication} {With} {Incremental} {Knowledge}-{Based} {HARQ}},
	volume = {3},
	issn = {2644-125X},
	doi = {10.1109/OJCOMS.2022.3189023},
	journal = {IEEE Open Journal of the Communications Society},
	author = {Zhou, Qingyang and Li, Rongpeng and Zhao, Zhifeng and Xiao, Yong and Zhang, Honggang},
	year = {2022},
	note = {Conference Name: IEEE Open Journal of the Communications Society},
	pages = {1076--1089},
}

@article{2023_wang_knowledgeEnhancedSemCom,
	title = {Knowledge {Enhanced} {Semantic} {Communication} {Receiver}},
	doi = {10.48550/arXiv.2302.07727},
	publisher = {arXiv},
	author = {Wang, Bingyan and Li, Rongpeng and Zhu, Jianhang and Zhao, Zhifeng and Zhang, Honggang},
	month = {2},
	year = {2023},
	note = {arXiv:2302.07727 [cs]},
}

@ARTICLE{2022_wang_transformerEmpowered6gMimoSemCom,
  author={Wang, Yang and Gao, Zhen and Zheng, Dezhi and Chen, Sheng and Gunduz, Deniz and Poor, H. Vincent},
  journal={IEEE Wireless Communications}, 
  title={{Transformer-Empowered 6G Intelligent Networks: From Massive MIMO Processing to Semantic Communication}}, 
  year={2022},
  volume={},
  number={},
  pages={1-9},
  doi={10.1109/MWC.008.2200157}}

@misc{2022_chaccour_LessDataMoreKnowledge,
  doi = {10.48550/ARXIV.2211.14343},
  author = {Chaccour, Christina and Saad, Walid and Debbah, Merouane and Han, Zhu and Poor, H. Vincent},
  title = {{Less Data, More Knowledge: Building Next Generation Semantic Communication Networks}},
  publisher = {arXiv},
  year = {2022},  
  copyright = {arXiv.org perpetual, non-exclusive license}
}

@article{2023_garcez_neurosymbolicAI3rdWave,
	title = {{Neurosymbolic {AI}: The 3rd Wave}},
	issn = {1573-7462},
	doi = {10.1007/s10462-023-10448-w},
	journal = {Artificial Intelligence Review},
	author = {Garcez, Artur d’Avila and Lamb, Luís C.},
	month = {3},
	year = {2023},
}

@INPROCEEDINGS{2022_guo_signalShapingSemComm,
  author={Guo, Shuaishuai and Wang, Yanghu and Zhangz, Peng},
  booktitle={2022 IEEE 96th Vehicular Technology Conference (VTC2022-Fall)}, 
  title={{Signal Shaping for Semantic Communication Systems with A Few Message Candidates}}, 
  year={2022},
  volume={},
  number={},
  pages={1-5},
  doi={10.1109/VTC2022-Fall57202.2022.10012981}}

@article{2021_ning_surveyMetaverse,
  author    = {Huansheng Ning and
               Hang Wang and
               Yujia Lin and
               Wenxi Wang and
               Sahraoui Dhelim and
               Fadi Farha and
               Jianguo Ding and
               Mahmoud Daneshmand},
  title     = {{A Survey on Metaverse: The State-of-the-art, Technologies, Applications,
               and Challenges}},
  journal   = {CoRR},
  volume    = {abs/2111.09673},
  year      = {2021},
  eprinttype = {arXiv},
  eprint    = {2111.09673},
  bibsource = {dblp computer science bibliography, https://dblp.org}
}

@article{2022_sun_surveyMetaverse,
  doi = {10.48550/ARXIV.2210.07990},
  url = {https://arxiv.org/abs/2210.07990},
  author = {Sun, Jiayi and Gan, Wensheng and Chao, Han-Chieh and Yu, Philip S.},
  title = {{Metaverse: Survey, Applications, Security, and Opportunities}},
  publisher = {arXiv},
  year = {2022},
  copyright = {arXiv.org perpetual, non-exclusive license}
}

@ARTICLE{2023_kshetri_economicsIndustrialMetaverse,
  author={Kshetri, Nir},
  journal={IT Professional}, 
  title={{The Economics of the Industrial Metaverse}}, 
  year={2023},
  volume={25},
  number={1},
  pages={84-88},
  doi={10.1109/MITP.2023.3236494}}

@Article{2022_albakri_vretReview,
AUTHOR = {Albakri, Ghaida and Bouaziz, Rahma and Alharthi, Wallaa and Kammoun, Slim and Al-Sarem, Mohammed and Saeed, Faisal and Hadwan, Mohammed},
TITLE = {{Phobia Exposure Therapy Using Virtual and Augmented Reality: A Systematic Review}},
JOURNAL = {Applied Sciences},
VOLUME = {12},
YEAR = {2022},
NUMBER = {3},
ARTICLE-NUMBER = {1672},
ISSN = {2076-3417},
DOI = {10.3390/app12031672}
}

@article{2018_eaton_specificPhobias,
	title = {{Specific Phobias}},
	volume = {5},
	issn = {2215-0366},
	doi = {10.1016/S2215-0366(18)30169-X},
	number = {8},
	journal = {The Lancet. Psychiatry},
	author = {Eaton, William W and Bienvenu, O Joseph and Miloyan, Beyon},
	month = {8},
	year = {2018},
	pmid = {30060873},
	pmcid = {PMC7233312},
	pages = {678--686},
}

@Article{2019_balan_fearLevelClassification,
AUTHOR = {Bălan, Oana and Moise, Gabriela and Moldoveanu, Alin and Leordeanu, Marius and Moldoveanu, Florica},
TITLE = {{Fear Level Classification Based on Emotional Dimensions and Machine Learning Techniques}},
JOURNAL = {Sensors},
VOLUME = {19},
YEAR = {2019},
NUMBER = {7},
ARTICLE-NUMBER = {1738},
PubMedID = {30978980},
ISSN = {1424-8220},
DOI = {10.3390/s19071738}
}

@INPROCEEDINGS{2022_kollias_learningSyntheticData,
  author={Kollias, Dimitrios},
  booktitle={European Conference on Computer Vision, 2022}, 
  title={{{ABAW}: Learning from Synthetic Data \& Multi-Task Learning Challenges}}, 
  year={2022},
  volume={},
  number={},
  doi={10.48550/arXiv.2207.01138}}

@article{2019_kollias_deepAffectPrediction,
	title = {Deep {Affect} {Prediction} in-the-{Wild}: {Aff}-{Wild} {Database} and {Challenge}, {Deep} {Architectures}, and {Beyond}},
	volume = {127},
	issn = {1573-1405},
	doi = {10.1007/s11263-019-01158-4},
	number = {6},
	journal = {International Journal of Computer Vision},
	author = {Kollias, Dimitrios and Tzirakis, Panagiotis and Nicolaou, Mihalis A. and Papaioannou, Athanasios and Zhao, Guoying and Schuller, Björn and Kotsia, Irene and Zafeiriou, Stefanos},
	month = {6},
	year = {2019},
	pages = {907--929},
}

@INPROCEEDINGS{2019_handrich_simultaneousValenceArousal,
  author={Handrich, Sebastian and Dinges, Laslo and Saxen, Frerk and Al-Hamadi, Ayoub and Wachmuth, Sven},
  booktitle={2019 IEEE International Conference on Signal and Image Processing Applications (ICSIPA)}, 
  title={{Simultaneous Prediction of Valence / Arousal and Emotion Categories in Real-Time}}, 
  year={2019},
  volume={},
  number={},
  pages={176-180},
  doi={10.1109/ICSIPA45851.2019.8977743}}

@article{2016_redmon_YOLO,
author    = {Joseph Redmon and
           Ali Farhadi},
title     = {{{YOLO9000:} Better, Faster, Stronger}},
journal   = {CoRR},
volume    = {abs/1612.08242},
year      = {2016},
eprinttype = {arXiv},
eprint    = {1612.08242},
bibsource = {dblp computer science bibliography, https://dblp.org}
}

@ARTICLE{2021_ieee802.11Standard,
  author={},
  journal={IEEE Std 802.11ax-2021 (Amendment to IEEE Std 802.11-2020)}, 
  title={{IEEE Standard for Information Technology--Telecommunications and Information Exchange between Systems Local and Metropolitan Area Networks--Specific Requirements Part 11: Wireless LAN Medium Access Control (MAC) and Physical Layer (PHY) Specifications Amendment 1: Enhancements for High-Efficiency WLAN}}, 
  year={2021},
  volume={},
  number={},
  pages={1-767},
  doi={10.1109/IEEESTD.2021.9442429}}

@ARTICLE{2023_gunduz_beyondBits,
  author={Gündüz, Deniz and Qin, Zhijin and Aguerri, Inaki Estella and Dhillon, Harpreet S. and Yang, Zhaohui and Yener, Aylin and Wong, Kai Kit and Chae, Chan-Byoung},
  journal={IEEE Journal on Selected Areas in Communications}, 
  title={Beyond Transmitting Bits: Context, Semantics, and Task-Oriented Communications}, 
  year={2023},
  volume={41},
  number={1},
  pages={5-41},
  doi={10.1109/JSAC.2022.3223408}}

@ARTICLE{2013_bengio_featureLearning,
  author={Bengio, Yoshua and Courville, Aaron and Vincent, Pascal},
  journal={IEEE Transactions on Pattern Analysis and Machine Intelligence}, 
  title={Representation Learning: A Review and New Perspectives}, 
  year={2013},
  volume={35},
  number={8},
  pages={1798-1828},
  doi={10.1109/TPAMI.2013.50}}

\end{document}